\documentclass[letterpaper, 10 pt,romanappendices]{IEEEtran}
\usepackage{textcomp}

\usepackage{amsmath,amssymb}

\usepackage{amsthm}
\usepackage{amsfonts,mathrsfs}
\usepackage{color}
\usepackage{graphicx}
\usepackage{epsfig}
\usepackage{subfigure}
\usepackage{enumerate}

\usepackage{algorithm,algorithmicx}
\usepackage{algpseudocode}
\usepackage[hidelinks]{hyperref}
\usepackage{marginnote}
\usepackage{empheq}
\usepackage{bm}
\usepackage{accents}
\usepackage{cite}
\usepackage{balance}
\usepackage{latexsym}
\usepackage{graphicx}
\usepackage{float}
\usepackage{colortbl}
\usepackage{hyperref}
\usepackage{enumerate}
\usepackage{multicol}
\usepackage{float}
\usepackage{cite}
\usepackage{marginnote}
\usepackage[normalem]{ulem} 

\usepackage[usenames,dvipsnames,svgnames,table]{xcolor}
\usepackage{tikz}
\usepackage{pgfplots}
\pgfplotsset{compat=newest} 
\pgfplotsset{plot coordinates/math parser=false} 
\usepgfplotslibrary{patchplots}
\usetikzlibrary{automata,positioning}

\usepackage{setspace}

\newcommand{\citep}{\cite}

\DeclareMathOperator{\proj}{proj}

\usepackage{soul}

\newcommand{\mc}{\mathcal}
\newcommand{\bb}{\mathbb}
\newcommand{\R}{\bb R}
\newcommand{\N}{\mathbb{N}}

\DeclareMathAlphabet{\mathbbmsl}{U}{bbm}{m}{sl}

\newcommand{\dom}{\operatorname{dom}}

\newcommand{\red}{\textcolor{red}}

\newcommand{\argmin}{\operatorname{argmin}}

\newcommand{\fix}{\operatorname{fix}}
\newcommand{\prox}{\mathrm{prox}}
\newcommand{\Id}{\mathrm{Id}}

\newcommand{\diag}{\operatorname{diag}}

\newcommand{\col}{\operatorname{col}}
\newcommand{\zer}{\operatorname{zer}}

\newcommand{\dist}{\mathrm{dist}}
\newcommand{\nc}{\mathrm{N}}
\newcommand{\0}{\mathbf{0}}
\newcommand{\1}{\mathbf{1}}

\newcommand{\gph}{\operatorname{gph}}

\newcommand{\Rmnum}[1]{\expandafter\@slowromancap\romannumeral #1@}
\newcommand{\eod}{\ensuremath{\hfill\Box}}
\newcommand{\qedd}{\ensuremath{\hfill \blacksquare}}

\newcommand{\ball}[2]{\{{#1}\} +{#2}\mathbb B}

\newcommand{\bomega}{\bm{\omega}}
\newcommand{\omegaopt}{\bomega^{\star}}

\newcommand{\btheta}{\boldsymbol{\theta}}

\newcommand{\innerprod}[2]{\left\langle {#1}, {#2} \right\rangle}

\usepackage{changepage}

\makeatletter
\newcommand{\specialcell}[1]{\ifmeasuring@#1\else\omit$\displaystyle#1$\ignorespaces\fi}
\makeatother
\newtheorem{assumption}{Assumption}
\newtheorem{proposition}{Proposition}
\newtheorem{corollary}{Corollary}
\newtheorem{definition}{Definition}
\newtheorem{lemma}{Lemma}
\newtheorem{example}{Example}
\newtheorem{theorem}{Theorem}{\it}{}

\newtheorem{remark}{Remark}

\newcommand{\VI}{\mathrm{VI}}

\newcommand{\comment}[1]{\textcolor{red}{[#1]}}

\title{Optimal {selection and tracking of} generalized Nash equilibria in monotone games }
\author{Emilio Benenati, Wicak Ananduta, and Sergio Grammatico
\thanks{The authors are with the Delft Center of Systems and Control (DCSC), TU Delft, the Netherlands. E-mail addresses: \texttt{\{e.benenati, w.ananduta, s.grammatico\}@tudelft.nl}. }
\thanks{This work was partially supported by the ERC under research project COSMOS (802348). }
}
\begin{document}
	
	\maketitle
\begin{abstract}
A fundamental open problem in monotone game theory is the computation of a  specific generalized Nash equilibrium (GNE) among all the available ones, e.g. the optimal equilibrium with respect to a system-level objective. The existing GNE seeking algorithms have in fact convergence guarantees toward an arbitrary, possibly inefficient, equilibrium. In this paper, we solve this open problem by leveraging results from fixed-point selection theory and in turn derive distributed algorithms for the computation of an optimal GNE in monotone games. 	We then extend the technical results to the time-varying setting and propose  an algorithm that tracks the sequence of optimal equilibria up to an asymptotic error, whose bound depends on the local computational capabilities of the agents.
\end{abstract}
	
	\begin{IEEEkeywords}
	 Multi-agent systems, Nash equilibrium seeking, Optimization
	\end{IEEEkeywords}
\vspace{-10pt}
	\section{Introduction}

\subsubsection*{Motivation}
Numerous engineering systems of recent interest, such as smart electrical grids {\cite{shilov21b, belgioioso21},} traffic {control} systems {\cite{ Bakhshayesh21},} and {wireless} communication systems {\cite{pang2008, Scutari2012, Wang14}} can be modelled as a \emph{generalized game}, that is, a system of multiple agents aiming at optimizing {their} individual, inter-dependent objectives, while satisfying some common constraints. A typical operating point for these systems is the Generalized Nash Equilibrium (GNE), where no agent can unilaterally improve their objective function \cite{facchinei10}. 
	\par The recent literature has witnessed the development of {theory and} algorithms for {computing} a \textit{variational} GNE (v-GNE) \cite{facchinei07a,facchinei10,kulkarni12}, which exhibits desirable properties of fairness and stability.  \textit{Semi-decentralized} {GNE seeking algorithms}, {where a reliable central coordinator gathers and broadcasts aggregate information,} have been proposed for strongly monotone \cite{paccagnan18, belgioioso18} and {merely} monotone games \cite{yin11, belgioioso17,belgioioso20}. A breakthrough idea in \cite{yi19}, later generalized for non-strongly monotone games {\cite{yi19b, franci20,belgioioso20distributed},} enables a \textit{distributed} computation of GNEs {by exploiting a suitable consensus protocol \cite{ bianchi22},} {thus requiring a peer-to-peer information exchange.}   
\par	 Existing results present, {however}, two fundamental shortcomings that might limit their practical application. {First,} unless strong assumptions are considered (namely, strong monotonicity of the pseudogradient), a game may {have} infinitely many v-GNEs \emph{and {the vast majority of the} existing algorithms provide no characterization of the equilibrium computed.}  {For instance, a Nash equilibrium can be arbitrarily inefficient with respect to system-level efficiency metrics (e.g., overall social cost) \cite{marden:14}.} {Such uncertainty on the obtained equilibrium is often unacceptable.} {A notable exception is the Tikhonov regularization algorithm \cite{yin11}, which  guarantees convergence to the minimum-norm solution. {In addition,} the method in \cite{dreves18, dreves19} seeks  the {(not necessarily variational)} GNE  closest to a desired strategy via a double-layer algorithm.} Second, decision-making agents often operate in a time-dependent environment and, due to the limited computation capabilities and  to the time required to  exchange information, it can be impossible to ensure a time-scale separation between the environment and the algorithm dynamics. This results in non-constant objectives and constraints between the {discrete-time} algorithmic iterations, as discussed in \cite{Simonetto2020}, {and the references therein,} for the particular case of optimization problems. Only few works, {e.g.,} \cite{lu2020,meng21},  consider this setting in the case of game {equilibrium problems} and only with a strong monotonicity assumption {on the game pseudogradient mapping.}

\subsubsection*{Optimal equilibrium selection and tracking}
{We can formulate the first issue, identified in the seminal work {\cite[Sect. 6]{facchinei10},} as} an \emph{optimal GNE selection} problem, that is, the problem of computing {a} GNE of a game (among the potentially infinitely many) that satisfies a selection criterion. This criterion {characterizes the desired equilibrium and} can be formalized as a system-level \emph{selection function} to be optimized over the set of GNEs. {For example,} {the system-level objective of an electricity market can be to {minimize} the deviation from an {efficient} operating set-point \cite{Simonetto2020}; for multiple autonomous vehicles, it can be to {minimize} the overall travel time of the network.}  {Meanwhile,} {the second} issue  can be cast as an \emph{optimal GNE tracking} problem, i.e., the problem of tracking the sequence of optimal GNEs of a time-varying game, with finite computation time and limited information on the future instances of the game {available.} As the GNE set is in general not a singleton, the tracking objective {should be} again chosen by means of a (time-varying) selection function. These problems, although of {high} practical interest, have never been addressed in the literature. 

{Under mild assumptions on the selection function,} the {optimal} GNE selection problem {in a monotone game} is a {special} case of a Variational Inequality ($\operatorname{VI}$) \cite{facchinei07} defined over the set of v-GNEs. {On the other hand,} as shown in \cite{franci20,belgioioso17,belgioioso20}, operator splitting techniques \cite{bauschke11} can be leveraged to characterize v-GNEs as the zeros of a monotone operator and, in turn, as the fixed-point set of a suitable operator. {Therefore, here we can cast} the problem as {that of} fixed-point selection \cite{yamada05}. {In the literature, e.g., {\cite{yamada05, xu03, cegielski14},} the latter} can be solved by the Hybrid Steepest Descent Method (HSDM), {whose iterations depend on the fixed-point operator (whose definition depends on the {primitives} of the game) and the monotone operator that defines the VI, namely the gradient of the selection function {in our setting}. \par

\subsubsection*{Contributions}
	In the first part of {the} paper {(Sections \ref{sec:op_gne} and \ref{sec:special_cases}),} we propose the first {distributed} algorithms for {solving} the optimal GNE selection {problem.} Our method employs the Forward-Backward-Forward (FBF) operator \cite{franci20} combined with the HSDM. We show that the proposed algorithm guarantees convergence to the \emph{optimal} v-GNE {set} in monotone games. Moreover, for {a special class} of monotone games, namely cocoercive games with affine coupling constraints, we also show that the preconditioned Forward-Backward (pFB) \cite{belgioioso17} can be paired with the HSDM to derive optimal GNE selection algorithms. Technically, {our contribution is to} {show that these operators {fulfill special properties that guarantee} the convergence of the HSDM toward the solution set of the corresponding fixed-point selection $\operatorname{VI}$.} {Compared to the methods in \cite{ yin11, dreves18, dreves19}, our proposed algorithms significantly  generalize the class of selection functions and, being single-layer, they provide  a considerable advantage in computational and communication burden compared to \cite{dreves18, dreves19}.}
\par In the second part of the paper {(Section \ref{sec:tv_GNEselect}),} we formalize the optimal GNE tracking problem {as a time-varying fixed-point selection problem.} {Thus,} {as} a solution {framework,} we propose the \emph{restarted HSDM}, {which adapts its operators when the problem changes.} In line with the results in the time-varying optimization literature \cite{Bastianello2020,Simonetto17}, we show convergence up to a tracking error which depends on the {problem data} and can be controlled by a suitable {tuning} of the algorithm parameters. {Similarly to the equilibrium selection problem, the restarted HSDM works with the aforementioned fixed-point operators to solve the optimal GNE tracking problem for the corresponding classes of monotone games.} 

\subsubsection*{Paper organization}
In Section \ref{sec:preliminaries}, we survey the required mathematical background and present a generalization {of the class of operators that comply with the conditions for applying the HSDM.} In Section \ref{sec:op_gne}, we formalize the optimal GNE selection problem and we {explain} our FBF-based algorithm  for general monotone games, {while} Section \ref{sec:special_cases} {discusses}  the pFB-based algorithm for {cocoercive} games. In Section \ref{sec:tv_GNEselect}, we formalize the optimal GNE tracking problem and we present the performance properties of the restarted HSDM algorithm. Finally, Section \ref{sec:illustrative_ex} illustrates the advantages of our methods on a peer-to-peer electricity market case study.

	\section{Mathematical preliminaries}
	\label{sec:preliminaries}
	
	\paragraph*{Notation} The set of real numbers is denoted by $\bb R$.   {The vector of all $1$ (or $0$) with dimension $n$ are denoted by $\1_n$ ($\bm 0_n$). We omit the subscript when the dimension is clear from the context.}   
	The operator $\col(\cdot)$ stacks the arguments column-wise. For a group of vectors $x_i$, $i \in \mc I =\{1,2,\dots,N\}$, we use the bold symbol to denote their column concatenation, i.e., $\bm x := \col(\{x_i\}_{i \in \mc I})$. The cardinality of a set is denoted by $|\cdot|$. The operator $\langle x, y \rangle$ denotes the inner product. We denote by $\|\cdot\|$ the Euclidean norm {and by $\|\cdot\|_p$ the $p$-norm.} 
	Let $P \succ 0$ be symmetric. For $x,y \in \bb R^n$,  $\langle x, y \rangle_P = \langle x, Py \rangle$  denotes the $P$-weighted Euclidean inner product. The graph of an operator $A \colon \bb R^n \rightrightarrows \bb R^n$ is denoted by $\gph (A)$. $\zer(A)$ defines the set of zeros of operator $A$, i.e., $\zer(A) := \{x \in \dom(A) \mid 0 \in A(x)\}$ whereas $\fix(A)$ defines the set of fixed points of $A \colon \bb R^n \to \bb R^n$, i.e., $\fix(A) := \{x \in \dom(A) \mid A(x)=x \}$.
	
	\paragraph*{Convex functions}
	A continuously differentiable function $f \colon \bb R^n \to \bb R$ is $\sigma$-strongly convex with respect to a $p$-norm, with $\sigma >0$, if, for all $x,x' \in \dom f$,
	$ f(x') \geq f(x) + \langle \nabla f(x), x'-x \rangle + \frac{\sigma}{2}\|x'-x\|_{p}^2.$
	Additionally, $f$ is convex if the previous inequality hold for $\sigma =0$. {The projection onto a closed convex set $C$ is denoted by $	\proj_C^\Psi(x) = \argmin_{z\in C}\| x-z\|_{\Psi}$, where $\Psi \succ 0$. For a convex function $f$ with subdifferential  $\partial f$ and $\Psi \succ 0$, the operator $\prox_{\partial f}^\Psi (x) := \argmin_z f(z) + \tfrac{1}{2}\|z-x \|_{\Psi}^2$ \cite[Def. 12.23]{bauschke11}. For example, for the indicator function of a closed convex set $C$, $\iota_C$, where $\partial \iota_C= \nc_C$ {being} the normal cone operator \cite[Ex. 1.25, 16.13]{bauschke11}, $\prox_{\iota_C}^\Psi(x)= \proj_C^\Psi(x)$ \cite[Ex. 12.25]{bauschke11}.  }
	
	\paragraph*{Operator theoretic definitions}
	An operator $A \colon \bb R^n \rightrightarrows \bb R^n$ is  \emph{monotone} if, for any $(x,y) \in \gph (A)$ and $(x',y') \in \gph (A)$,
	$ \langle y - y', x-x' \rangle \geq 0 $ \cite[Def. 20.1]{bauschke11}, 
	and $\beta$-\emph{strongly monotone} if $A-\beta \Id$, {where $\Id$ is the identity operator,} is monotone. 
	Let $C$ be a nonempty subset of $\bb R^n$. A single-valued operator $\mc T \colon C \to \bb R^n$ is \emph{Lipschitz continuous} if there exists a constant $L>0$, such that, for all $x, x' \in \bb R^n$, 
	$ \| \mc T(x) - \mc T(x')\| \leq L \|x-x'\|$ {\cite[Def. 1.47]{bauschke11}}. In particular, the operator  $\mc T$ is (i) \emph{nonexpansive} if {$L=1$}, (ii)  \emph{attracting nonexpansive} if $\mc T$ is nonexpansive with $\fix(\mc T)\neq \varnothing$ and $\|\mc T(x)-z\| < \| x - z\|$, for all $z \in \fix(\mc T)$ and all $x \notin \fix(\mc T)$; and (iii) \emph{quasi-nonexpansive} if $\fix(\mc T)\neq \varnothing$ and $\| \mc T(x)-z\| \leq \|x -z \|$, for all $z \in \fix(\mc T)$ and $x \in \bb R^n$. 
	Moreover, $\mc T$ is $\alpha$-averaged nonexpansive, for $\alpha \in (0,1)$, if there exists a nonexpansive operator $\mc R \colon C \to \bb R^n$ such that $\mc T = (1-\alpha)\Id + \alpha \mc R$.  {If $\mc T$ is averaged nonexpansive with $\fix(\mc T) \neq \varnothing$, then $\mc T$ is attracting \cite[Sec. 2.A]{yamada05}. Additionally $\mc T$ is $\beta$-\emph{cocoercive} if $ \langle \mc T(x) - \mc T(y), x-y\rangle \geq \beta \| \mc T(x) - \mc T(y)\|$.
	
	Now, let $C$ be a non-empty, closed, and convex subset of $\bb R^n$, $\mc T\colon \bb R^n \to \bb R^n$ be quasi-nonexpansive under the $\Psi$-induced norm $\|\cdot\|_{\Psi}$ for some positive definite matrix $\Psi$, i.e.,
	$\| \mc T(x)-z\|_\Psi \leq \|x -z \|_\Psi$, for all $z \in \fix(\mc T) \neq \varnothing$ and $x \in \bb R^n$. 
	We define the distance of a point $x \in \bb R^n$ to $C$ by {
	$	\dist_\Psi(x,C) := \inf_{z \in C} \|x-z\|_{\Psi},
	$. 
	}  
	For $r \geq 0$, we define the set 
	\begin{align}
		C_{\geq r}^\Psi := \{x \in \bb R^n \mid \dist_\Psi(x,C) \geq r \}. \label{eq:Cgeq}
	\end{align}
	Furthermore, let us define the function 
	\begin{equation}
		D_\Psi (r) \hspace{-2pt}:= \hspace{-2pt} \begin{cases}
			{\inf} \ \ \left( \dist_\Psi(x,\fix(\mc T))  - \dist_\Psi(\mc T(x),\fix(\mc T)) \right), \\
			\operatorname{s.t. } \  {x \in (\fix(\mc T))_{\geq r}^\Psi \cap C},
			\  \text{if } (\fix(\mc T))_{\geq r}^\Psi \hspace{-2pt}\cap\hspace{-2pt} C\hspace{-2pt} \neq\hspace{-2pt} \varnothing, \\
			+\infty, \quad \text{otherwise.}
		\end{cases}
		\label{eq:D}
	\end{equation}
	For $\Psi = I$, we omit the subscript of $D$. We sometimes refer to $D_{\Psi}$ as the \emph{shrinkage function} under the norm $\|\cdot\|_{\Psi}$. 
		
	The function $D_\Psi$ has the properties stated {next} in Proposition \ref{prop:D} {(see \cite[Prop. 2.6]{Cegielski13} for the case $\Psi=I$).} 
	\begin{proposition}
		\label{prop:D}
		{Let $\Psi$ be positive definite.} 
		For the function $D_\Psi$ defined in \eqref{eq:D}, it holds that:
		\begin{enumerate}[(i)]
			\item {$D_\Psi$ is positive semidefinite and non-decreasing;}
			\item $D_\Psi(\dist(x,\fix(\mc T))) \leq \| x - \mc T(x) \|_\Psi$ for all $x\in C$.
		\end{enumerate}
	\end{proposition}
	
	\begin{definition}[Quasi-shrinking \cite{yamada05}]
		\label{def:q_shrinking1}
		A quasi-nonexpansive operator $\mc T \colon \bb R^n \to \bb R^n$ is \emph{quasi-shrinking} on a non-empty, closed, and convex set $C \subseteq \bb R^n$ if $\fix(\mc T)\cap C \neq \varnothing$ and $D(r)= 0 \Leftrightarrow r=0$, where $D(r)$ is defined as in \eqref{eq:D}. 
	\end{definition}
	\begin{remark}
		\label{rem:q_shrinking}
		Suppose that a quasi-nonexpansive operator $\mc T$ is quasi-shrinking on C, i.e., $D(r) = 0 \Leftrightarrow r=0$. Then, it also holds that $D_\Psi(r)=0 \Leftrightarrow r=0$, for any   $\Psi \succ 0$. 
	\end{remark}
	
	\begin{example} 
		The Euclidean projection onto $C$, $\proj_C$ is quasi-shrinking and its shrinkage function (defined in (\ref{eq:D})) is
		{$$D(r)=\inf_{\{u|\dist(u, C)\geq r\}} \dist(u, C) - \underbrace{\dist(\proj_C(u),C)}_{=0} = r .$$}
	\end{example}

{Finally, we identify a class of quasi-shrinking operators, as formally stated in Lemma \ref{le:demi_closedness_quasi_shrinking}, which generalizes the result in \cite[Prop. 2.11]{Cegielski13} and is useful for our analysis.}
\begin{definition}[Demiclosed operator {\cite[Def. 4.26]{bauschke11}}]
	Let $C\subseteq\R^{{n}}$ {be a} closed {set.} An operator $\mc T\colon C\to \R^{{n}}  $ is demiclosed {at} $u\in\R^{{n}}$ if $\mc T(\bomega^{\infty})=u$, for any sequence $(\bomega_k)_{k\in\N}\in C$  such that $\lim_{k\xrightarrow[]{}\infty}\bomega_k=\bomega^{\infty}$ and $\lim_{k\xrightarrow[]{}\infty}\mc T (\bomega_k)=u$. 
\end{definition}	
\begin{lemma}
\label{le:demi_closedness_quasi_shrinking}
    Let $\mc T$ be quasi-nonexpansive, with $\fix(\mc T)\neq \varnothing$. Let ${\mc T_2}$ be an operator such that ${\Id}-\mc T_2$ is demiclosed {at} $0$ and such that $\fix({\mc T_2})\subseteq \fix(\mc T) $. {Assume that} for any $\omegaopt\in\fix(\mc T)$, 
    \begin{equation}
    \label{eq:condition_quasi_shrinking}
        \| \mc T(\bomega)-\omegaopt \|^2_{\Psi} \leq \| \bomega-\omegaopt \|^2_{\Psi} - \gamma \| \bomega - \mc T_2(\bomega) \|^2_{\Psi},
    \end{equation}
	{for some $\gamma>0$ and $\Psi \succ 0$.} 
    Then, $\mc T$ is quasi-shrinking on any compact {convex} set $C$ such that $C\bigcap\fix(\mc T)\neq \varnothing$. 
\end{lemma}

\begin{proof}
	See Appendix \ref{appendix:proof:le:demi_closedness_quasi_shrinking}.
\end{proof}

	\section{Optimal selection of generalized Nash equilibria}
	\label{sec:op_gne}
	\subsection{Generalized Nash equilibrium problem}
	\label{sec:gnep}
	
	
	Let {us consider} $N$ agents, denoted by the set $\mc I := \{1, 2, \dots, N\}$, {with} inter-dependent optimization problems: 
	\begin{subequations}
		\begin{empheq}[left={\forall i \in \mc I \colon \empheqlbrace\,}]{align}
			\underset{x_i \in \mc X_i}{\operatorname{min}} \quad & J_i(\bm x):= {\ell_i(x_i)  + } f_i(\bm x) \label{eq:cost_f} \\
			\operatorname{s.t.} \quad & \sum_{j \in \mc I} g_j(x_j) \leq 0, \label{eq:coup_const}
		\end{empheq}
		\label{eq:gen_game}%
	\end{subequations}
	where $x_i \in \bb R^{n_i}$ is the decision variable of agent $i$ whereas $\bm x := \col((x_i)_{i \in \mc I}) \in \bb R^n$ is a concatenated vector of the decision variables of all agents. {Let us} use $\bm x_{-i}= \col(\{x_j\}_{j \in \mc I\backslash \{i\}})$ to denote the concatenated decision variables of all agents except agent $i$.   Let $\mc X_i \subseteq \bb R^{n_i}$ denote the local feasible set of $x_i$ and $J_i \colon \bb R^n \to \bb R$ denote the cost function of agent $i$ that depends on the decision variables of other agents. 
	Moreover, \eqref{eq:coup_const} {represents} a separable coupling constraint  with the function $g_j \colon \bb R^{n_j} \to \bb R^{m}$ associated with agent $j$.
	
	
	We denote the collective feasible set of {the game in} \eqref{eq:gen_game} by  
		{ \begin{equation}
				{\bm{\Omega} := \prod_{i \in \mc I} \mc X_i \cap \Big\{\bm x \mid \sum_{j \in \mc I} g_j(x_j) \leq 0  \Big\}.}\label{eq:X}
		\end{equation}}%
		Here, we look for equilibrium solutions to \eqref{eq:gen_game} where no agent has the incentive to unilaterally deviate, namely, GNE:
	\begin{definition}
		A set of strategies $\bm x^\ast := \col((x_i^\ast)_{i \in \mc I})$ is a generalized Nash equilibrium (GNE) of the game in \eqref{eq:gen_game} if $\bm x^\ast \in {\bm{\Omega}}$ and, for each $i \in \mc I$,
		\begin{equation}
			J_i(\bm x^\ast) \leq  J_i(x_i,\bm x_{-i}^\ast),
		\end{equation}
		for any $x_i \in \mc X_i \cap \{y \mid g_i(y) \leq - \sum_{j \in \mc I \backslash \{i\} } g_j(x_j^\ast) \}$. \eod 
	\end{definition}

	Furthermore, we focus on the class of  jointly convex GNEP and {hence,} consider the following assumptions on Problem \eqref{eq:gen_game} \cite[{Assms}  1--2]{belgioioso17}. We note that  {\cite{yi19,yi19b,franci20,belgioioso20distributed,belgioioso20}} consider the case {of affine constraint functions.}
	\begin{assumption}
		\label{as:gen_game}
		{In \eqref{eq:gen_game},} for each $i \in \mc I$, {the functions}  {$f_i(\cdot,\bm x_{-i})$,} for any $\bm x_{-i}$, and $g_i(\cdot)$ are {component-wise} convex and continuously differentiable; $\ell_i$ is convex and lower semicontinuous. {For each $i \in \mc I$, the set} $\mc X_i$   is nonempty, compact, and convex. The global feasible set { $\bm{\Omega}$ defined in \eqref{eq:X}} is non-empty and satisfies Slater's constraint qualification {\cite[Eq. (27.50)]{bauschke11}.} 
	\end{assumption}
	\begin{assumption}
		\label{as:pseudograd}
		The  mapping 
		{\begin{equation}
			F(\bm x)	:=\col(({\nabla_{x_i} f_i} (\bm x))_{i\in\mc I}), \label{eq:pseudograd}
		\end{equation}}
		 {with $f_i$ as in \eqref{eq:cost_f},} is monotone. 
	\end{assumption}

 {As in {\cite{belgioioso17,yi19,yi19b,franci20,belgioioso20distributed,belgioioso20},}} we can formulate the problem of finding a GNE of the game in \eqref{eq:gen_game} as that of a monotone inclusion. To this end, we introduce the dual variable $\lambda_i \in \bb R^{m}$, {for each $i \in \mc I$,} to be associated with the coupling constraint \eqref{eq:coup_const}. 
	{{Furthermore,} we focus on a subset of GNEs, namely variational GNE (v-GNE), indicated by  equal optimal dual variables, $\lambda_i^\ast = \lambda^\ast$, for all $i \in \mc I$. As discussed in \cite{facchinei10,kulkarni12}, a v-GNE enjoys several desirable properties, such as   fairness and larger social stability than non-variational ones.  Under Assumptions \ref{as:gen_game}--\ref{as:pseudograd}, the set of v-GNEs of the game in \eqref{eq:gen_game} {is non-empty} \cite[Prop. 12.11]{palomar10}. The Karush-Kuhn-Tucker (KKT) optimality conditions of a v-GNE of the game in \eqref{eq:gen_game}, denoted by $\bm x^\ast$, {are}:
	\begin{subequations}
		\begin{empheq}[left={\hspace{-5pt} \forall i \in \mc I \colon \empheqlbrace\,}]{align}
			&	\0 \in \nc_{\mc X_i} ( x_i^\ast)+ \partial_{x_i} J_i(\bm x^\ast) + \langle \nabla g_i(x_i^\ast), {\lambda}^\ast  \rangle, \label{eq:KKT_1}\\
			&	 \0 \in \nc_{\bb R^m_{\geq 0}}({\lambda}^\ast) - \sum_{j \in \mc I} g_j(x_j^\ast).\label{eq:KKT_2}
		\end{empheq}
		\label{eq:KKT}%
	\end{subequations}
} 
	
	To obtain a v-GNE via a fully distributed algorithm, we incorporate a consensus scheme {on} the dual variables. In the full information case, one typically assumes that there exists a communication network over which the agents exchange information to update their dual variables. Let us represent this communication network  as an undirected graph $\mc G^\lambda = (\mc I, \mc E^\lambda)$ and assume that $\mc G^\lambda$ is connected. Furthermore, we denote the Laplacian of $\mc G^\lambda$ by $L$ and the neighbors of agent $i$ in $\mc G^\lambda$ by $\mc N_i^\lambda$, i.e., $\mc N_i^\lambda := \{j \in \mc I \mid (i,j) \in \mc E^\lambda \}$. Additionally, let $\mc N_i^J$ denote the set of agents whose decision variable $x_j$ influences the cost function $J_i$. For simplicity, we assume that $\mc N_i^J \subseteq \mc N_i^\lambda$. 
	
	Now, let us denote $\nu_i \in \bb R^m$ as the consensus variable of agent $i$, and $\bm \omega =(\bm x, \bm \lambda, \bm \nu) \in \bb R^{n_\omega}$, where $\bm \lambda = \col(\{\lambda_i\}_{i \in \mc I})$,  $\bm \nu = \col(\{\nu_i\}_{i \in \mc I})$, and $n_\omega = n+2Nm$. Then, we can define the operators $\mc A\colon {\bm{\mc X}} \times \bb R^{Nm}_{\geq 0} \times \bb R^{Nm} \to \bb R^{n_\omega} $, $\mc B\colon \bb R^{n_\omega} \to \bb R^{n_\omega}$, and $\mc C \colon \bb R^{n}\times \bb R_{\geq 0}^{Nm} \times \bb R^{Nm} \to \bb R^{n_\omega}$, as follows:
	{
		\begin{align}
		\mc A(\bm \omega)  
		&:= \prod_{i \in \mc I} (\nc_{\mc X_i} {+ \partial \ell_i} )(x_i) \times \nc_{\bb R^{Nm}_{\geq 0}}(\bm \lambda) \times \{\0_{Nm}\}, \label{eq:op_A}\\
		\mc B(\bm \omega) 
		&:= \col({F(\bm x)}, (L \otimes I_m) \bm \lambda, \0_{Nm} ), \label{eq:op_B}\\
		\mc C(\bm \omega) 
		&:= \col((\langle \nabla g_i(x_i), \lambda_i \rangle)_{i \in \mc I}, - (g_i(x_i))_{i \in \mc I} - (L \otimes I_m)\bm \nu, \notag \\
		& \qquad \quad (L \otimes I_m) \bm \lambda  ). \label{eq:op_C}
	\end{align} }%
	In turn, we can translate the GNEP {in \eqref{eq:gen_game}} as a monotone inclusion problem, i.e., 
	\begin{equation}
		\text{find }\bm \omega \text{ such that } {\bm \omega \in \zer(\mc A + \mc B + \mc C).} \label{eq:mon_incl}
	\end{equation}
	Similarly to \cite[Thm. 2]{yi19}, we can show that for any $\bm \omega$ such that \eqref{eq:mon_incl} holds, we obtain the pair $(\bm x,\lambda)$  that satisfies the KKT conditions in \eqref{eq:KKT} if Assumptions \ref{as:gen_game}-\ref{as:pseudograd} hold (see Appendix \ref{ap:prop_ABC} for details).  Furthermore, due to the maximal monotonicity of $(\mc A + \mc B + \mc C)$ (Lemma \ref{le:max_mon_ABC} in Appendix \ref{ap:prop_ABC}), $\zer(\mc A + \mc B + \mc C)$ is convex \cite[Prop. 23.39]{bauschke11}. Additionally, since the set of v-GNE of the game is bounded {as it is a subset of $\bm{\mc X}$,} the set of solutions of the inclusion in \eqref{eq:KKT} and {the set} $\zer(\mc A + \mc B + \mc C)$ are  bounded \cite[Prop. 3.3]{auslender00}. 
	
	\subsection{Optimal equilibrium selection problem}
	{The inclusion problem in \eqref{eq:mon_incl} {may have multiple} solutions. In this section, we want to find {{an} equilibrium} solution  that minimizes a selection function, denoted by } 
	 $\phi \colon \bb R^{n_\omega} \to \bb R$, i.e., 
		\begin{equation}
		\left\{
		\begin{array}{rl}
			\underset{\bm \omega}{\argmin}  & \phi(\bm \omega)  \\		
			\operatorname{s.t.}  & \bm \omega \in \zer(\mc A+\mc B+\mc C).
		\end{array}	
		\right.
		\label{eq:gne_opt1}%
	\end{equation}
 
	For example, we can consider {the selection function}
	\begin{equation}
		{\phi_{\mathrm{ex}}}(\bm \omega) = \| Q \bm \omega - \bm \omega^{\mathrm{ref}}\|, \label{eq:psi_ex}
	\end{equation}
	for some $Q \succcurlyeq 0$. When $Q = I$ and $\omega^{\mathrm{ref}}=\0$, the objective is to find a minimum norm v-GNE. The vector $\omega^{\mathrm{ref}}$ can be any desired strategy of the agents, and thus the objective is to find {the} v-GNE closest to this strategy, {as discussed in \cite{dreves18,dreves19}.}  In some engineering applications, such as electrical networks, \eqref{eq:psi_ex} can represent system level objectives (see Section \ref{sec:illustrative_ex}).
	{In} the remainder of the paper, we consider the following {technical} assumption on the selection function, {which, together with the convexity of $\zer(\mc A + \mc B + \mc C)$, guarantees that the {optimization} problem in \eqref{eq:gne_opt1} is convex.}
	\begin{assumption}
		\label{as:phi}
		The  function $\phi$ {in \eqref{eq:gne_opt1}} is  
		continuously differentiable,  convex, {and has $L_{\phi}$-Lipschitz continuous gradient}. 
	\end{assumption}

	As {a} first step {towards} computing an optimal variational GNE, we 
 leverage existing results to derive operators $\mc T$ with the property {that}
	\begin{equation}
	\label{eq:equivalence_zer_fix}
	    \bomega \in \zer (\mc A + \mc B + \mc C) \Leftrightarrow \bomega \in \fix(\mc T),
	\end{equation}
{{and such that {the Banach-Picard iteration of $\mc T$  \cite[Sect. 5.2]{bauschke11}} guarantees convergence to a solution of the inclusion in \eqref{eq:mon_incl}.} For instance, for cocoercive generalized games, a preconditioned forward-backward (pFB) {operator} {presents the desired characteristics}\cite{yi19}, 
whereas the {forward-reflected-backward (FRB) operator \cite{malitsky20}} or the forward-backward-forward (FBF) {operator} \cite{tseng00} meets {these requirements} even for general monotone games. Furthermore, {we require that the operator $\mc T$ in \eqref{eq:equivalence_zer_fix} can be evaluated} in a distributed manner.} 
{Therefore, by \eqref{eq:equivalence_zer_fix} {and Assumption \ref{as:phi},} the optimal equilibrium selection problem in \eqref{eq:gne_opt1} can be cast as a fixed-point selection $\operatorname{VI}$:} 
	\begin{equation}
		\text{find } \bm \omega^\star \text{ s.t. } \inf_{\bm \omega \in \fix(\mc T)} \innerprod{\bm \omega - \bm \omega^\star}{{\nabla \phi} (\bm \omega^\star)} \geq 0. \label{eq:gne_opt}
	\end{equation}
	\subsection{Distributed optimal equilibrium selection algorithm}
		\label{sec:distributed_GNE_selection}
	{With the aim of solving the $\operatorname{VI}$ in \eqref{eq:gne_opt},} we consider a fixed-point selection algorithm called the hybrid steepest descent method (HSDM)  \cite{yamada05}, which is defined by the following {discrete-time dynamical system} or iteration:
	\begin{equation}
		\label{eq:hsdm}
		\bomega^{(k+1)}=\mc T(\bomega^{(k)})-\beta^{(k)} {\nabla \phi }(\mc T(\bomega^{(k)})).
	\end{equation} 
	The HSDM can solve Problem  \eqref{eq:gne_opt} when $\mc T$ is quasi-nonexpansive and quasi-shrinking  with bounded $\fix(\mc T)$, as formally stated next.
	\begin{assumption}
		\label{as:beta}
		The step size of the HSDM $\beta^{(k)}$ satisfies:
		\begin{enumerate}[(i)]
			\item \label{as:beta_nonsummable} $\lim_{k\to \infty} \beta^{(k)} = 0$, $\sum_{k\geq 1} \beta^{(k)} =  \infty$;
			\item \label{as:beta_squaresummable} $\sum_{k\geq 1} (\beta^{(k)})^2 < \infty$. \eod 
		\end{enumerate}
		\end{assumption}
		\begin{remark}
			{The sequence} $\beta^{(k)}= {\beta_0}/{k^p}$, for any $\beta_0 >0$ and $p \in ({1/2}, 1]$, satisfies Assumption \ref{as:beta}. \eod 
		\end{remark}
		\begin{assumption}
		\label{as:quasi_nonexp}
		    $\mc T$ is quasi-nonexpansive. \eod
		\end{assumption}
		\begin{assumption}
		\label{as:quasi_shrinking}
		   There exists a nonempty bounded closed convex set $C$ on which $\mc T$ is quasi-shrinking. \eod
		\end{assumption}

		\begin{lemma}[{From} {\cite[Thm. 5]{yamada05}}]
		\label{th:hsdm_cvx}
		Let Assumption \ref{as:phi} hold and {$\Omega^\star$} be the set of solutions of  {the $\operatorname{VI}$ in \eqref{eq:gne_opt},} with non-empty and bounded $\fix(\mc T)$. Suppose that $\mc T$ satisfies Assumptions \ref{as:quasi_nonexp} {and} \ref{as:quasi_shrinking} and that  $(\bm{\omega}^{(k)})_{k \geq 0} \subset C$. If the step size $\beta^{(k)}$ satisfies Assumption \ref{as:beta}.\ref{as:beta_nonsummable}, then the HSDM in \eqref{eq:hsdm} generates {a} sequence $(\bm \omega^{(k)})_{k \in \bb N}$ such that $$ \lim_{k\to \infty} \dist(\bm \omega^{(k)},\Omega^\star) = 0. $$ 
		\end{lemma}
	
	Therefore,   our main {technical} task {is to find} a suitable operator $\mc T$ {that can be evaluated} in a distributed manner and {that} satisfies both \eqref{eq:equivalence_zer_fix} and Assumptions \ref{as:quasi_nonexp}--\ref{as:quasi_shrinking}, required for the convergence of the HSDM sequence. 


	
	{ {Under mere} monotonicity of the pseudogradient mapping (Assumption \ref{as:pseudograd}), perhaps the most obvious choice is the FRB splitting, which, however, is not quasi-nonexpansive\footnote{{The FRB {iteration} does not generate a Fej\'{e}r monotone sequence \cite[Prop. 2.3]{malitsky20}, implying that it is not quasi-nonxepansive {and violates Assumption \ref{as:quasi_nonexp}.}}} (and, thus, {it} is not quasi-shrinking).}  
{Another viable option is the FBF splitting method \cite{tseng00}, which works for v-GNE seeking {in} monotone games satisfying Assumptions \ref{as:gen_game}--\ref{as:pseudograd}, as shown in  \cite{belgioioso17,franci20}.} {As our first {technical} result, we show that the FBF algorithm satisfies both the desired property in (\ref{eq:equivalence_zer_fix}) and  Assumptions \ref{as:quasi_nonexp}--\ref{as:quasi_shrinking}. {To that end,}	firstly, we {compactly} state the FBF operator for \eqref{eq:mon_incl}, as follows:
	\begin{align}
		\mc T_{\mathrm{FBF}}( \bm \omega)
		& := ((\Id- \Psi^{-1}(\mc B + \mc C) )(\Id +{\Psi^{-1}\mc A})^{-1} \notag \\
		&\quad \cdot (\Id- \Psi^{-1}(\mc B + \mc C)) + \Psi^{-1}(\mc B + \mc C))(\bm \omega),
		\label{eq:T_FBF}
	\end{align}
	where $\Psi \succ 0$ is a diagonal positive definite matrix. The FBF requires the forward operator, which is $(\mc B + \mc C)$, to be Lipschitz continuous. A sufficient condition for this requirement is given in Assumption \ref{as:Lips} (see Lemma \ref{le:lips_BC} in Appendix \ref{ap:prop_ABC}). Under maximal monotonicity and Lipschitz continuity, it holds that $\zer(\mc A + \mc B + \mc C) = \fix(\mc T_{\mathrm{FBF}})$ (see Lemma \ref{le:fix_FBF} in Appendix \ref{ap:le:fix_FBF}). In addition, by denoting $L_B$ as the Lipschitz constant of $\mc B+\mc C$, we define the step-size  matrix $\Psi$ in Assumption \ref{as:stepsizeFBF}, which guarantees the convergence of the sequence generated by the fixed-point iteration with $\mc T_{\mathrm{FBF}}$ toward a point in $\zer(\mc A + \mc B + \mc C)$.
	
	\begin{assumption}
		\label{as:Lips}
		The mapping {$F(\bm x)$} in \eqref{eq:pseudograd}
		is $L_F$-Lipschitz continuous. Furthermore, for each $i \in \mc I$, the  function  {$g_i$ in \eqref{eq:coup_const}} has a  bounded and {$L_{ g}$-Lipschitz} continuous gradient. 
	\end{assumption}
	\begin{assumption}
		\label{as:stepsizeFBF}
		{It holds that $ |\Psi^{-1}| \leq 1/L_B$,} where $L_B>0$ {is} the Lipschitz constant of $\mc B + \mc C$ and $ \Psi = \diag(\rho^{-1}, \sigma^{-1}, \tau^{-1})$, where $\rho = \diag(\{\rho_i I_{n_i}\}_{i \in \mc I})$, $\tau = \diag(\{\tau_i I_{m}\}_{i \in \mc I})$, and $\sigma = \diag(\{\sigma_i I_{m}\}_{i \in \mc I})$.  
	\end{assumption}
	
	We are now ready to present the distributed FBF for seeking an optimal variational GNE based on the selection function $\phi(\bm \omega)$ via the HSDM as shown in Algorithm \ref{alg:fbf}. 




	{To have a convergence guarantee as stated in Lemma \ref{th:hsdm_cvx}, the FBF operator must satisfy Assumptions \ref{as:quasi_nonexp} and \ref{as:quasi_shrinking}.  {Let us} show that this is the case in the following lemma.

	\begin{lemma}
		\label{le:FBF_q_nonexp}
		Let Assumptions \ref{as:gen_game}, \ref{as:pseudograd}, and \ref{as:Lips} hold. The operator  $\mc T_{\mathrm{FBF}}$ in \eqref{eq:T_FBF}, where $\mc A$, $\mc B$, and $\mc C$ are defined in \eqref{eq:op_A}--\eqref{eq:op_C} and $\Psi$ is defined in Assumption \ref{as:stepsizeFBF}, is quasi-nonexpansive and quasi-shrinking {on any compact convex set $C$ such that $C \cap \fix(\mc T_{\mathrm{FBF}}) \neq \varnothing$  (see} Definition \ref{def:q_shrinking1}). 
	\end{lemma} 
	\begin{proof}
		See Appendix \ref{pf:le:FBF_q_nonexp}.
	\end{proof}
	Furthermore, we observe that the HSDM sequence generated by using $\mc T_{\mathrm{FBF}}$ is bounded, as formally stated next.

	\begin{lemma}
		\label{le:bounded_HSDM} 
		Let Assumptions {\ref{as:gen_game}--\ref{as:beta} and \ref{as:Lips}--\ref{as:stepsizeFBF}} hold. Then, the sequence $(\bm \omega^{(k)})_{k \in \bb N}$ generated by the HSDM method in \eqref{eq:hsdm} with  $\mc T =\mc T_{\mathrm{FBF}}$ in \eqref{eq:T_FBF}, where $\mc A$, $\mc B$, and $\mc C$ are defined in \eqref{eq:op_A}--\eqref{eq:op_C} and $\Psi$ is defined in Assumption \ref{as:stepsizeFBF}, is bounded, i.e., for any arbitrary $\bm \omega^\star \in \fix(\mc T_{\mathrm{FBF}})$, it holds that
		$	\|\bm \omega^{(k)} - \bm \omega^\star \| \leq R(\bm \omega^\star),$ 
		for some positive finite $R(\bm \omega^\star)$.
	\end{lemma}
		\begin{proof}
			See Appendix \ref{pf:le:bounded_HSDM}.
		\end{proof}
		
{{\begin{algorithm}[t]
	\caption{Optimal  v-GNE selection via FBF and HSDM}
	\label{alg:fbf}
	\textbf{Initialization.} Set $x_i^{(0)} \in \mc X_i$, $\lambda_i^{(0)} \in \bb R_{\geq 0}^m$, and $\nu_i^{(0)} \in \bb R^m$, for all $i \in \mc I$.
	
	\textbf{Iteration of each agent $i \in \mc I$.}
	\begin{enumerate}
		\item Receives $x_j^{(k)}$ from agent {$j \in \mc N_i^J$} and $\lambda_j^{(k)}, \nu_j^{(k)}$ from agent $j \in \mc N_i^\lambda$. 
		\item Updates:
		{\small 	\begin{align*}
			\tilde x_i^{(k)} \hspace{-2pt} &=\hspace{-1pt}  {\prox_{\ell_i + \iota_{\mc X_i}}^{\rho_i}}\hspace{-5pt}\left(x_i^{(k)} \hspace{-3pt}- \hspace{-2pt} \rho_i \big(\nabla_{x_i}{f}_i(\bm x^{(k)}) \hspace{-2pt} + \hspace{-2pt} \nabla g_i(x_i^{(k)})^\top \hspace{-1pt} \lambda_i^{(k)}\big)\hspace{-1pt}\right)\hspace{-2pt}, \\
			\tilde \lambda_i^{(k)} &= \proj_{\geq 0}\Big(\Big.\lambda_i^{(k)} + \tau_i\Big(\Big. g_i(x_i^{(k)}) \\ & \qquad \quad \qquad\Big.\Big. +\sum_{j \in \mc N_i^\lambda} \left(\nu_i^{(k)}- \nu_j^{(k)} - \lambda_i^{(k)} + \lambda_j^{(k)} \right) \Big) \Big), \\
			\tilde \nu_i^{(k)} &= \nu_i^{(k)} - \sigma_i\textstyle\sum_{j \in \mc N_i^\lambda} \left(\lambda_i^{(k)}- \lambda_j^{(k)} \right).
		\end{align*}
		}
		\item Receives $\tilde x_j^{(k)}$ from agent $j \in \mc N_i^J$ and $\tilde \lambda_j^{(k)}, \tilde \nu_j^{(k)}$ from agent $j \in \mc N_i^\lambda$.
		\item Updates: 
		\begin{align}
			\accentset{\circ}	 x_i^{(k)} &= \tilde x_j^{(k)} - \rho_i \Big(\Big. \nabla_{x_i}f_i( \tilde{\bm x}^{(k)}) -\nabla_{x_i}f_i(\bm x^{(k)}) \notag \\
			& \quad +\nabla g_i(\tilde x_i^{(k)})^\top \tilde \lambda_i^{(k)}- \nabla g_i(x_i^{(k)})^\top \lambda_i^{(k)} \Big.\Big), \notag \\	
				\accentset{\circ} \lambda_i^{(k)} &= \tilde \lambda_i^{(k)} + \tau_i \Big(\Big. g_i(\tilde{x}_i^{(k)}) - g_i(x_i^{(k)}) \notag  \\
			& \quad  +\textstyle\sum_{j \in \mc N_i^\lambda}\Big(\Big. \tilde \nu_i^{(k)} - \nu_i^{(k)} - \tilde \nu_j^{(k)} + \nu_j^{(k)} \Big.\Big) \notag\\
			& \quad - \textstyle\sum_{j \in \mc N_i^\lambda}\Big(\Big. \tilde \lambda_i^{(k)} - \lambda_i^{(k)} - \tilde \lambda_j^{(k)} + \lambda_j^{(k)} \Big.\Big) \Big.\Big), \notag\\
			\accentset{\circ} \nu_i^{(k)} &= \tilde \nu_i^{(k)} -  \sigma_i  \textstyle\sum_{j \in \mc N_i^\lambda} \left( \tilde \lambda_i^{(k)} - \lambda_i^{(k)} - \tilde \lambda_j^{(k)} + \lambda_j^{(k)} \right).  \notag
		\end{align}
		\item Sends $(\accentset{\circ} x_i^{(k)}, \accentset{\circ}\lambda_i^{(k)}, \accentset{\circ}\nu_i^{(k)})$ to a coordinator and receives back $\nabla_{\omega_i} \phi( \accentset{\circ} {\bm x}^{(k)}, \accentset{\circ}{\bm \lambda}^{(k)}, \accentset{\circ}{\bm \nu}^{(k)}) $, where $\omega_i = (x_i, \lambda_i, \nu_i)$.
		\item Updates: 
		\begin{align}
					&\hspace{-10pt}	(x_i^{(k+1)}, \lambda_i^{(k+1)}, \nu_i^{(k+1)}) \notag \\
					&\hspace{-10pt}	=( \accentset{\circ} x_i^{(k)}, \accentset{\circ}\lambda_i^{(k)}, \accentset{\circ}\nu_i^{(k)}) - \beta^{(k)} \nabla_{\omega_i} \phi( \accentset{\circ} {\bm x}^{(k)}, \accentset{\circ}{\bm \lambda}^{(k)}, \accentset{\circ}{\bm \nu}^{(k)}). \label{eq:hsdm_fbf}
		\end{align}
	\end{enumerate}
\end{algorithm} }
	
{Thus}, we can {now} show that Algorithm \ref{alg:fbf} generates a sequence that converges toward the solution set of the {optimal} GNE selection problem {in} \eqref{eq:gne_opt}.
	\begin{theorem}
		\label{prop:FBF}
		Let Assumptions {\ref{as:gen_game}--\ref{as:beta} and \ref{as:Lips}--\ref{as:stepsizeFBF}} hold. Let $\Omega^\star$ be the set of solutions {to} Problem {\eqref{eq:gne_opt}} with $\mc T = \mc T_{\mathrm{FBF}}$ defined in \eqref{eq:T_FBF}, where $\mc A$, $\mc B$, and $\mc C$ are defined in \eqref{eq:op_A}--\eqref{eq:op_C}. Furthermore, let $(\bm \omega^{(k)})_{k \in \bb N}$, where $\bm \omega^{(k)}=(\bm x^{(k)}, \bm \lambda^{(k)}, \bm \nu^{(k)})$, be the sequence generated by Algorithm \ref{alg:fbf}. Then, 
		{$ \lim_{k\to \infty} \dist(\bm \omega^{(k)},\Omega^\star) = 0,$}
		and $(\bm x^{(k)})_{k \in \bb N}$ converges to {an} optimal v-GNE of the game in \eqref{eq:gen_game}. 
	\end{theorem}
	\begin{proof}
		See Appendix \ref{pf:prop:FBF}.
	\end{proof}

	{\begin{remark}
		A central coordinator and step 5 of Algorithm \ref{alg:fbf} are not needed if $\phi$ is a separable function, i.e., $\phi(\bm \omega) = \sum_{i\in \mc I} \phi_i(\omega_i)$. In this case, step 6 can be immediately executed by using local information $(\accentset{\circ} x_i^{(k)}, \accentset{\circ}\lambda_i^{(k)}, \accentset{\circ}\nu_i^{(k)})$ {only,} as long as each agent $i$ {knows} {the gradient} $\nabla \phi_i$. 
	\end{remark}}


\section{Optimal equilibrium selection in cocoercive games}
\label{sec:special_cases}
{
In this section, we discuss {a} special class of monotone games, namely cocoercive games with affine coupling constraints. {These games arise as a generalization of the widely studied class of strongly monotone games \cite{yi19},\cite{belgioioso18}. Differently from the strong monotonicity assumption, however, cocoercivity alone does not guarantee the uniqueness of the v-GNE. } 
	\begin{assumption}[{\cite[{Assm.} 5]{belgioioso20}}]
		\label{as:pseudograd_cocoercive}
		The   {mapping $F$ in \eqref{eq:pseudograd}} is $\eta$-cocoercive. 
	\end{assumption}
	\begin{assumption}[{{\cite[Eq. (3)]{belgioioso20}}}]
		\label{as:coup_const_lin}
		For each $i \in \mc I$, the function $g_i$ in \eqref{eq:coup_const} is affine, i.e., $g_i(x_i):= A_i x_i - b_i$, for some matrix $A_i \in \bb R^{m \times n_i}$ and vector $b_i \in \bb R^m$. 
	\end{assumption}

	For this particular class of games, the preconditioned forward-backward (pFB) splitting \cite{yi19} can efficiently compute a variational GNE. We note that, {although} \cite{yi19} considers games with strongly monotone pseudogradient, the FB splitting only requires cocoercivity of the forward operator{\cite[Thm. 26.14]{bauschke11}.}} {Compared with the FBF, the pFB has the advantages of only having one communication round per iteration (as opposed to two) and {larger step size bounds}. A {numerical} performance comparison is provided in \cite{franci20}.}

	Given the particular structure of the coupling constraint as stated in Assumption \ref{as:coup_const_lin}, we can rewrite the operators in \eqref{eq:mon_incl} as follows:
	\begin{align}
		\mc A(\bm \omega) &:= \prod_{i \in \mc I} (\nc_{\mc X_i} {+ \partial \ell_i} )(x_i)  \times \nc_{\bb R^{Nm}_{\geq 0}}(\bm \lambda) \times \{\0_{Nm}\}, \label{eq:op_A2}\\
		\mc B(\bm \omega) &:= \col({F(\bm x)}, (L \otimes I_m) \bm \lambda+\bm b, \0_{Nm} ), \label{eq:op_B2}\\
		\mc C(\bm \omega) &:= \col(\bm A^\top \bm \lambda, -  \bm A \bm x - (L \otimes I_m)\bm \nu, (L \otimes I_m) \bm \lambda  ), \label{eq:op_C2}
	\end{align} 
	where $\bm A = \diag(\{A_i\}_{i \in \mc I})$ and $\bm b = \col(\{b_i\}_{i \in \mc I})$. Thus, 
	the pFB operator for the monotone inclusion in \eqref{eq:mon_incl} based on the operators $\mc A$, $\mc B$, and $\mc C$  in \eqref{eq:op_A2}--\eqref{eq:op_C2} is given by \cite[Eq. (24)]{yi19}:
	\begin{equation}
		\mc T_{\mathrm{pFB}} (\bm \omega) :=  (\Id + \Phi^{-1} (\mc A + \mc C))^{-1}(\Id - \Phi^{-1}\mc B)(\bm \omega),
		\label{eq:T_pFB}
	\end{equation}
	where $\Phi \succ 0$ is a symmetric positive definite preconditioning matrix, defined as
	\begin{equation*}
		\Phi := \begin{bmatrix}
			\rho^{-1}  & -\bm A^\top & 0\\
			-\bm A  & \tau^{-1} & -L \otimes I_m \\
			0 & - L \otimes I_m & \sigma^{-1}
		\end{bmatrix},
	\end{equation*} 
	{with} $\rho,\theta,\tau \in \bb R_{>0}^N$ {being} step sizes similarly defined as those of the FBF algorithm. Then, we can have an extension of the pFB for the v-GNE optimal selection of cocoercive games, as stated in Algorithm \ref{alg:pFB} {with  step size rules} given in {Assumptions \ref{as:beta} and} \ref{as:stepsizepFB}. {Finally,} {we formally state the convergence property of Algorithm \ref{alg:pFB} in Theorem \ref{prop:pFB}.}

	\begin{assumption}[\hspace{0.5pt}{\cite[Eq. (27) and Thm. 3]{yi19}}]
		\label{as:stepsizepFB}
		It holds that $\rho_i \leq (\max_{j=1,\dots,n_i} \sum_{k=1}^m |[A_i]_{jk}| \delta)^{-1}$, $\tau_i  \leq (\max_{j=1,\dots,n_i} \sum_{k=1}^m |[A_i]_{jk}| + 2|\mc N_i^\lambda|+\delta)^{-1}$, and $\sigma_i \leq (2|\mc N_i^\lambda|+\delta)^{-1}$, for all $i \in \mc I$, where $\delta > 1/(\min (\eta, (2 \max_{i \in \mc I} |\mc N_i^\lambda| )^{-1}) )$. 
	\end{assumption}
	\begin{theorem} 
		\label{prop:pFB}
		Let Assumptions \ref{as:gen_game}--\ref{as:beta}, \ref{as:Lips}, and \ref{as:pseudograd_cocoercive}--\ref{as:stepsizepFB}  hold. Let $\Omega^\star$ be the set of solutions {to} Problem \eqref{eq:gne_opt} with $\mc T = \mc T_{\mathrm{pFB}}$ defined in \eqref{eq:T_pFB}, where $\mc A$, $\mc B$, and $\mc C$ are defined in \eqref{eq:op_A2}--\eqref{eq:op_C2}. Furthermore, let $(\bm \omega^{(k)})_{k \in \bb N}$, where $\bm \omega^{(k)}=(\bm x^{(k)}, \bm \lambda^{(k)}, \bm \nu^{(k)})$, be the sequence generated by Algorithm \ref{alg:pFB}. Then, 
		$ \lim_{k\to \infty} \dist(\bm \omega^{(k)},\Omega^\star) = 0,$
		and $(\bm x^{(k)})_{k \in \bb N}$ converges to {an} optimal v-GNE of the game in \eqref{eq:gen_game}. 
	\end{theorem}
	\begin{algorithm}
	\caption{\color{black} Optimal v-GNE selection  via pFB and HDSM for linearly coupled cocoercive games}
	\label{alg:pFB}
	\textbf{Initialization.} Set $x_i^{(0)} \in \mc X_i$, $\lambda_i^{(0)} \in \bb R_{\geq 0}^m$, and $\nu_i^{(0)} \in \bb R^m$, for all $i \in \mc I$.
	
	\textbf{Iteration of each agent $i \in \mc I$.}
	\begin{enumerate}
		\item Receives $x_j^{(k)}$ from agent $j \in \mc N_i^J$ and $\lambda_j^{(k)}$ from agent $j \in \mc N_i^\lambda$. 
		\item Updates:
		\begin{align*}
			\accentset{\circ} x_i^{(k)} &= {\prox_{\ell_i + \iota_{\mc X_i}}^{\rho_i}}\left(x_i^{(k)} -  \rho_i (\nabla_{x_i}{f_i}(\bm x^{(k)}) + A_i^\top \lambda_i^{(k)})\right), \\
			\accentset{\circ} \nu_i^{(k)} &= \nu_i^{(k)} - \sigma_i\textstyle\sum_{j \in \mc N_i^\lambda} \left(\lambda_i^{(k)}- \lambda_j^{(k)} \right).
		\end{align*}
		\item Receives $\tilde \nu_j^{(k)}$ from agent $j \in \mc N_i^\lambda$. 
		\item Updates:
		\begin{align}
			&\accentset{\circ} \lambda_i^{(k)} = \proj_{\geq 0}\Big(\Big.\lambda_i^{(k)} + \tau_i\Big(\Big. A_i(2x_i^{(k+1)}-x_i^{(k)}) - b_i  \notag \\
			&\qquad \quad +\textstyle\sum_{j \in \mc N_i^\lambda} \left(2 \nu_i^{(k+1)}- 2\nu_j^{(k+1)} - \nu_i^{(k)} + \nu_j^{(k)}  \right) \notag \\
			&\qquad \quad +\textstyle\sum_{j \in \mc N_i^\lambda} \left( \lambda_i^{(k)} - \lambda_j^{(k)} \right) \Big.\Big) \Big.\Big). \notag 
		\end{align}
		\item Sends $(\accentset{\circ} x_i^{(k)}, \accentset{\circ}\lambda_i^{(k)}, \accentset{\circ}\nu_i^{(k)})$ to a coordinator and receives back $\nabla_{\omega_i} \phi( \accentset{\circ} {\bm x}^{(k)}, \accentset{\circ}{\bm \lambda}^{(k)}, \accentset{\circ}{\bm \nu}^{(k)}) $, where $\omega_i = (x_i, \lambda_i, \nu_i)$.
		\item Updates: 
		\begin{align}
			&\hspace{-10pt}(x_i^{(k+1)}, \lambda_i^{(k+1)}, \nu_i^{(k+1)}) \notag \\
			&\hspace{-10pt}=( \accentset{\circ} x_i^{(k)}, \accentset{\circ}\lambda_i^{(k)}, \accentset{\circ}\nu_i^{(k)}) - \beta^{(k)} \nabla_{\omega_i} \phi( \accentset{\circ} {\bm x}^{(k)}, \accentset{\circ}{\bm \lambda}^{(k)}, \accentset{\circ}{\bm \nu}^{(k)}). \label{eq:hsdm_pfb}
		\end{align}
	\end{enumerate}
\end{algorithm}

	\begin{proof}
		See Appendix \ref{pf:prop:pFB}.
	\end{proof}

\medskip
	\section{Online {tracking of} optimal generalized Nash equilibria}
	\label{sec:tv_GNEselect}
	    \subsection{{Online optimal equilibrium tracking problem}}
	    \label{sec:introduction_of_time_varying_GNE_sel}
	    
	    In the second part of this paper, we {consider} the online GNE selection problem. 	Specifically, let us {introduce} the time-varying game: 
	    \begin{subequations}

	    	\begin{empheq}[left={{\forall t\in\N, \forall i \in \mc I} \colon \empheqlbrace\,}]{align}
	    		\underset{x_i \in \mc X_{i,t}}{\min} \quad & J_{i,t}(\bm x) \label{eq:cost_f_time_var} \\
	    		\operatorname{s.t.} \quad & \sum_{j \in \mc I} g_{j,t}(x_j) \leq 0, \label{eq:coup_const_time_var}
	    	\end{empheq}
	    	\label{eq:time_varying_game}%
	    \end{subequations}
    	{where $t$ denotes the time index.}  
	    {The problem is time-varying in the sense that the objective functions of the agents, as well as the constraints, may vary over time. We assume that each instance of the games {in \eqref{eq:time_varying_game}} satisfies Assumptions \ref{as:gen_game} and \ref{as:pseudograd}.}  The time-varying GNE selection problem {thus} concerns the tracking of the sequence $(\omegaopt_t)_{t\in\N}$:
	    \vspace{-2pt}
	    \begin{subequations}
	    	\begin{empheq}[left={\forall t\in\N \colon~\omegaopt_t :=  \empheqlbrace\,}]{align}
	    	\underset{\bomega}{\argmin} ~& \phi_t(\bomega) \\
	    	\operatorname{s.t.} ~&  \bomega \in \zer(\mc A_t + \mc B_t + \mc C_t).
	        \end{empheq}
	        \label{eq:online_GNE_selection}%
	    \end{subequations} 
	    {The problems in \eqref{eq:time_varying_game} and  \eqref{eq:online_GNE_selection} are a sequence in time of instances of \eqref{eq:gen_game} and \eqref{eq:gne_opt1}, respectively.} The operators $\mc A_t$, $\mc B_t$, {and} $\mc C_t$ are defined in \eqref{eq:op_A}--\eqref{eq:op_C}, for the game in \eqref{eq:time_varying_game} at time step $t$. The agents need to compute the action $\bomega_{t+1}$, having only access to the game formulation up to time $t$. This setup describes the case in which the agents act in a variable environment with limited computation capabilities, so that they cannot compute the exact optimal selection before changes in the problem (either in the selection function or in the game)  occur. \par
	    
	  For every $t {\in \bb N}$, and under a suitable choice of operator $\mc T_t$, such that \vspace{-2pt}
	   $$ \bomega \in \zer (\mc A_t + \mc B_t + \mc C_t) \Leftrightarrow \bomega \in \fix(\mc T_t),$$ 
	   $\omegaopt_t$ \eqref{eq:online_GNE_selection} can be equivalently found as
     the solution of the {time-varying} fixed-point selection problem
    	\begin{equation}
    	\label{eq:time_varying_VI}
    	    \inf_{\bomega \in \fix(\mc T_t)} \langle \bomega-\omegaopt_t, \nabla\phi_{{t}}(\omegaopt_t) \rangle \geq 0.
    	\end{equation}}%
     {The sequence} $(\omegaopt_t)_{t\in\N}$ {is} well defined {when, for each $t \in \N$,} the solution of \eqref{eq:online_GNE_selection} {is} unique. {Let us} then introduce the following assumptions, which guarantee {uniqueness.}}
     
	   	\begin{assumption} 
	   	\label{as:phi_time_varying}
		The selection function $\phi_t \colon \bb R^{n_\omega} \to {\bb R}$ {in \eqref{eq:time_varying_VI}} is  
		continuously differentiable,  $\sigma$-strongly convex, and has $L_{\phi}$-Lipschitz continuous gradient for all $t \in \N$. 
		\end{assumption} 
		\begin{assumption} 
			\label{as:quasi_nonexp_always}
			{The operator} $\mc T_t$ {in \eqref{eq:time_varying_VI}} is quasi nonexpansive with $ \fix(\mc T_t)\neq \varnothing$ for all $t \in \N$. 
		\end{assumption}
	   { Under Assumptions  \ref{as:phi_time_varying} and \ref{as:quasi_nonexp_always}, by \cite[Prop. 1]{yamada05}, we find $\fix(\mc T_t)$ to be closed and convex for all $t$. By \cite[Thm. 2.3.3]{facchinei07}, the problem in (\ref{eq:time_varying_VI}) has a unique solution for all $t$. These assumptions also guarantee, by \cite[Thm. 2A.7]{Dontchev14}, that the solutions of \eqref{eq:time_varying_VI} coincide with the solutions of \eqref{eq:online_GNE_selection}.} In the remainder of this section, we build {upon} the results of Section \ref{sec:distributed_GNE_selection} to derive an HSDM-inspired algorithm for tracking $(\omegaopt_t)_{t\in\N}$.
	
	\begin{subsection}{Online fixed point tracking via the {restarted} Hybrid Steepest Descent Method}
		The existing results on the HSDM algorithm study the asymptotic behavior with vanishing step size $(\beta^{(k)})_{k\in\N}$ {(see Assumption \ref{as:beta}).} However, in online scenarios,  decision makers {may} not have the computational capability to exactly compute the fixed point of the algorithm, {since} that would require an infinite amount of iterations in a limited time span before a new instance of the problem becomes available. {Thus, we propose} and study the (approximate) convergence properties of an algorithm that only performs a finite number of HSDM iterations per time step. {Consequently,} the sequence of step sizes becomes truncated and a sequence of vanishing step sizes, which is required for the convergence of the HSDM, cannot be defined. We therefore simplify the analysis by considering a constant sequence of step sizes.

	
		Let us introduce the restarted HSDM algorithm. 
		{Given an initial state} $\bomega_1$, {for each $t\in \N$, {we propose the following:}}
		\begin{align}
				\bm y^{(k+1)}&:=
				\begin{cases}
					 \bomega_t, &  \text{for} ~ k=1,\\
					\mc T_t(\bm y^{(k)}) - \beta \nabla \phi_t( \mc T_t(\bm y^{(k)})), & \text{for}~k=2,..., K,
				\end{cases} \notag \\
				\bomega_{t+1}& := \bm y^{(K+1)}.
				\label{eq:algorithm_online_HSDM}
		\end{align} 
	    In words, at each time step $t$ the auxiliary variable $\bm y^{(k)}$, with $k=1, ...K$, is updated with $K$ iterations of the HSDM. Then, the decision variable at time step $t+1$ is obtained as $\bm \omega_{t+1}=\bm y^{(K+1)}$. The algorithm is then restarted when the information on the selection function and game for the next time step {becomes} available. Next, let us postulate the following {technical} assumptions:
   
\begin{assumption}
\label{as:bounded_optimizer_sequence}
	There exists a compact set $\mc Y$ such that
		    $\omegaopt_t \in  \mc Y$ for all $t \in \N.$
\end{assumption}	   
\begin{assumption}
\label{as:bounded_gradient}
There exists $U\geq0$ such that 
{$\sup_{\boldsymbol{\omega} \in \bigcup_{{\tau} \in\N} \mathrm{Im}(\mc T_{{\tau}}), t\in\N}\|{\nabla}\phi_t(\bomega)\|\leq U. $}
\end{assumption}	   

Assumption \ref{as:bounded_optimizer_sequence} is {practically} reasonable, {since we can assume that we do not aim at tracking a divergent sequence.} Assumption \ref{as:bounded_gradient} {specifies an upper bound for the gradient of the selection function} and is in line with the online optimization literature ({see} \cite[Assm. 5]{zinkevich03}, \cite[Assm. 5]{DallAnese2016}, {among others).}

As shown in Section \ref{sec:distributed_GNE_selection}, the HSDM method converges to the solution of a selection problem over the fixed point set of a quasi-shrinking operator. {In the online scenario, assuming the operator $\mc T_t$ to be quasi-shrinking for all $t$ is not enough, as the  quasi-shrinking property might not hold asymptotically. {Thus,} we {also} {postulate} the following {technical} assumption: 

\begin{assumption} (Uniformly quasi-shrinking operator)
\label{as:T_uniformly_quasi_shrinking}
For any closed convex set $C$ such that $C \cap \fix(\mc T_t)\neq \varnothing$, there exists $D:\R_{\geq 0}\to \R$ positive semidefinite such that $D_t(r)\geq D(r)$ for all $t {\in \N}$ and for all $r\geq0$, where $D_t(\cdot)$ is the shrinkage function of $\mc T_t$ defined as in \eqref{eq:D}.  
\end{assumption} 
\begin{remark}
    Assumption \ref{as:T_uniformly_quasi_shrinking} implies that $\mc T_t$ is quasi-shrinking on any closed, convex set $C$ such that $\mc C \cap \fix(\mc T_t) \neq \varnothing$, $ \forall t \in \N$. 
\end{remark}
}

		The next lemma outlines a contraction property of the restarted HSDM to the {solution sequence of Problem \eqref{eq:time_varying_VI}} up to an additive error, which can be controlled by an appropriate choice of {the step size} $\beta$ and {the number of iterations} $K$.
		
	\begin{lemma}
			\label{lemma:finite_hsdm_bound_single_iteration} 
			Let Assumptions \ref{as:phi_time_varying}--\ref{as:T_uniformly_quasi_shrinking} hold. For any $t\in\N$, let $\bomega_{t+1}$ be generated by the restarted HSDM algorithm in (\ref{eq:algorithm_online_HSDM}). For any $\gamma>0$, there exist $K, \beta >0$, such that 	
			\begin{equation}
				\label{eq:contractive-like_bound}
				\|\bm \omega_{t+1}-\bm \omegaopt_t\|^2 \leq \left(1-{\tau(\beta)}\right)^K \| \bm \omega_{t}-\bm \omegaopt_t \|^2 + \gamma,
			\end{equation}
			where
			$ \tau(\beta):= 1- \sqrt{1-\beta(2\sigma-\beta L_{\phi}^2)} \in (0,1).$ 
		\end{lemma}
		\begin{proof}
		    See Appendix \ref{appendix:proof:lemma:finite_hsdm_bound_single_iteration}.
		\end{proof}
		\begin{remark}
			\label{remark:approximation_error_control_1}
        	For decreasing values of the {tolerable error} $\gamma$, the stepsize $\beta$ has to be decreased and the number of iterations $K$ has to be increased {(see the proof of Lemma \ref{lemma:finite_hsdm_bound_single_iteration}).} 
		\end{remark}
			We now proceed to show how the property in (\ref{eq:contractive-like_bound}) can be exploited to derive an error bound on {the trajectory tracking of the solution sequence of the problem in \eqref{eq:time_varying_VI} via the restarted HSDM  \eqref{eq:algorithm_online_HSDM}.}  {Thus,}  we introduce the following assumption:
			{
			\begin{assumption}
				\label{ass:var_lim}
				There exist scalars $\delta_1, \delta_2 \geq0 $  such that
				\begin{enumerate}[(i)]
					\item \label{ass:variability_limited}  $\sup_{t \in \N}\| \bm \omegaopt_{t+1}-\bm \omegaopt_{t} \| \leq \delta_1$;
					\item \label{ass:GNE_set_bounded_variations} 
					$\sup_{t \in \N} \dist(\omegaopt_t, (\fix(\mc T_{t+1})) \leq \delta_2$. 
				\end{enumerate} 
				  
			\end{assumption}
			}
			
 Assumption {\ref{ass:var_lim}.\ref{ass:variability_limited} }is standard in online optimization (e.g. \cite[Assm. 1]{Simonetto2020}, \cite[Assm. 3.1]{Simonetto17}, and \cite[Assm. 3]{DallAnese2016}). { We note that Assumption {\ref{ass:var_lim}.\ref{ass:variability_limited}} implies Assumption {\ref{ass:var_lim}.\ref{ass:GNE_set_bounded_variations}.}} The latter is nevertheless introduced to distinguish the effects of the time variation of $\mc T_t$ (which influences both $\delta_1$ and $\delta_2$) {from} the one of $\phi_t$ (which only influences $\delta_1$). 
            \begin{remark}
               {If} $\mc T_t=\mc T$, for all $t\in\N$, and the time dependence can be expressed through a parametrization, that is,   $\phi_t(\bomega)=\phi(\bomega,t)$, then an estimate for $\delta_1$ can be found.   In fact, if $\phi(\bomega,t)$ is continuously differentiable, we find by \cite[Thm. 2F.7]{Dontchev14} that the solution mapping, that is, the mapping from $t$ to the solution of $\VI(\fix(\mc T), \nabla_{x}\phi(\cdot, t) ) $, is Lipschitz continuous in a neighbourhood of any ${t}$ with Lipschitz constant $\sigma^{-1}|\nabla_t \phi(\omegaopt_{{t}},{t}) |$. Thus, if the time variation between two consecutive {time steps} $t_1$ and $t_2$ is small enough, $\delta_1$ can be estimated as $\sigma^{-1}|\nabla_t \phi(\omegaopt_{t_1},{t_1}) |(t_2-t_1)$. The solution mapping is in general discontinuous when $\mc T_t$ is time-varying; thus, a similar estimate cannot be found in the general case.
            \end{remark}

			\begin{theorem}
				\label{prop:asymptotic_limit_online_hsdm}
				Let Assumptions \ref{as:phi_time_varying}--\ref{ass:var_lim} hold. Let the sequence $(\bm \omega_t)_{t\in\N}$ be generated {by} the restarted HSDM in \eqref{eq:algorithm_online_HSDM}. 
				{ For any $\gamma>0$, there exist $\beta \in(0,\frac{2\sigma}{L^2_{{\phi}}})$ and $\bar{K}$, such that, for all $K\geq \bar K$, the sequence $(\bomega_t)_{t\in\N}$ is bounded and
				\begin{equation}
					\label{eq:asymptotic_tracking_error}
					\limsup_{t\rightarrow\infty} \| \bm \omega_{t}-\bm \omegaopt_{t}\|^2 \leq  \frac{({\gamma}+\delta_1^2)}{1/2-\alpha}, 
				\end{equation}
				where $\alpha=(1-\tau(\beta))^K<\frac{1}{2}$.
				} 
			\end{theorem}
			\begin{proof}
			    See Appendix \ref{appendix:proof:prop:asymptotic_limit_online_hsdm}.
			\end{proof}
			
	
			\begin{remark}
				In {Theorem} \ref{prop:asymptotic_limit_online_hsdm},  $\gamma$ is derived from the additive error in \eqref{eq:contractive-like_bound}. Thus, to control the approximation error in (\ref{eq:asymptotic_tracking_error}), $\beta$ {must} be chosen small {so} to obtain small values of ${\gamma}$,  as pointed out in Remark \ref{remark:approximation_error_control_1}. However, the value $\tau(\beta)$ tends to $0$ for small values of $\beta$. This leads to the denominator in (\ref{eq:asymptotic_tracking_error}) to be  small for small stepsizes, unless the number of iterations $K$ is increased. Therefore, a smaller step size leads to a better approximation error only if it is shouldered by an increase in the number of iterations of the algorithm per time step. 
			\end{remark}
			
		{In summary,} we find that the restarted HSDM \eqref{eq:algorithm_online_HSDM} asymptotically tracks the solutions trajectory of the online fixed point selection problem in \eqref{eq:time_varying_VI}, with an asymptotic error  that can be controlled up to the variability of the problem $\delta_1$, via an appropriate choice of  $\beta$, $K$, {as shown in {Theorem} \ref{prop:asymptotic_limit_online_hsdm}. }
			Additionally, we {emphasize that the results hold for a more general problem, i.e., one could replace $\nabla \phi_t$ in Problem \eqref{eq:time_varying_VI} with a strongly monotone operator to obtain} an extension to the fixed-point selection problem in \cite{yamada05}. In the next section, we use the restarted HSDM to solve the online GNE tracking problem {in}  \eqref{eq:online_GNE_selection}. 
		\subsection{ Distributed {optimal} equilibrium tracking algorithm for monotone games}
		\label{sec:online_GNE_selection}
		We recall from Section \ref{sec:distributed_GNE_selection}  that the set of variational GNEs for a monotone game  can be characterized as the set of fixed points  of the operator $ \mc T_{\mathrm{FBF}}$ defined in (\ref{eq:T_FBF}). {Thus, for the time-varying game in \eqref{eq:time_varying_game} at time $t$, let $\mc T_{\mathrm{FBF},t}$ be the FBF operator defined as:
			\vspace{-5pt}
		\begin{align}
			\mc T_{\mathrm{FBF},t}( \bm \omega)
			& := ((\Id- \Psi^{-1}(\mc B_t + \mc C_t) )(\Id +{\Psi^{-1}\mc A_t})^{-1} \notag \\
			&\quad \cdot (\Id- \Psi^{-1}(\mc B_t + \mc C_t)) + \Psi^{-1}(\mc B_t+ \mc C_t))(\bm \omega),
			\label{eq:T_FBF_t}
		\end{align}
	where $\mc A_t$, $\mc B_t$, and $\mc C_t$ are those in Problem \eqref{eq:online_GNE_selection} and associated with the game in \eqref{eq:time_varying_game} at time $t$.}  
	 The solutions of the time-varying GNE selection problem in $\eqref{eq:online_GNE_selection}$ {are equivalent to the solutions of \eqref{eq:time_varying_VI},} with $\mc T_t=\mc T_{\mathrm{FBF},t}$ for all $t$. By Lemma \ref{le:FBF_q_nonexp}, $\mc T_{\mathrm{FBF},t}$, for each $t$, is a quasi-nonexpansive, quasi-shrinking operator. Therefore, the restarted HSDM algorithm in \eqref{eq:algorithm_online_HSDM} can be employed for tracking the solution trajectory, with an asymptotic tracking error given by {Theorem} \ref{prop:asymptotic_limit_online_hsdm}. We introduce {an assumption for} the GNE selection problem, which is equivalent to Assumption \ref{ass:var_lim}.\ref{ass:GNE_set_bounded_variations}:
		
		{\begin{assumption} 
				\label{ass:zero_set_bounded_variations}
				{There exists a scalar $\delta_2 \geq 0$ such that}
				$ \sup_{t\in\N} \dist(\omegaopt_t, \mathrm{zer}(\mc A_{t+1} + \mc B_{t+1} + \mc C_{t+1})) \leq \delta_2$.  
			\end{assumption} }
		\begin{corollary}
		\label{cor:online_FBF}
			{Let us consider the online GNE tracking problem {in} \eqref{eq:online_GNE_selection} for the time-varying game {in} \eqref{eq:time_varying_game} that satisfies Assumptions \ref{as:gen_game}, \ref{as:pseudograd}, \ref{as:Lips}, for each $t\in\N$.  Suppose that Assumptions \ref{as:phi_time_varying}, \ref{as:bounded_optimizer_sequence}, \ref{as:bounded_gradient}, \ref{ass:var_lim}, {\ref{ass:zero_set_bounded_variations}} hold. Let $\mc T_t = \mc T_{\mathrm{FBF},t}$ satisfy Assumption \ref{as:T_uniformly_quasi_shrinking}. Then, for any $\gamma>0$ there exist $\beta\in(0,\frac{2\sigma}{L_{\phi}^2})$ and $\bar{K}$ such that, for any $K\geq \bar{K}$, the asymptotic tracking {error} of Algorithm \ref{alg:Online_fbf} is given by \eqref{eq:asymptotic_tracking_error}.	}	\end{corollary}
\begin{proof}
	See Appendix \ref{app:pf:cor:online_FBF}.
\end{proof}

		\begin{remark}
		\label{rem:finite_set_operators}
		    In Corollary \ref{cor:online_FBF}, Assumption \ref{as:T_uniformly_quasi_shrinking} is satisfied for example when at every time step $t$, the feasible set of Problem \eqref{eq:online_GNE_selection} is selected among the GNE sets of {finitely many} games. That is, consider a finite set of operators
		    $$ \mc A_h, \mc B_h, \mc C_h, ~~ \text{with}~ h\in\{1, ..., H\}, $$
		    and for each $h$, the associated FBF operator $\mc T_{\mathrm{FBF}}^h$.
		    Defining a mapping from the time step $t$ to the indexes of the operators
		    $\eta:\N\to \{1, ..., H\}, $
		    Problem \eqref{eq:online_GNE_selection} is defined by
		    \begin{equation*}
	    	\omegaopt_t :=  \left\{
	    	\begin{array}{rl}
	    		&	\argmin_{\bomega} \phi_t(\bomega)\\
	    		&	\text{s.t.} ~~ \bomega \in \zer(\mc A_{\eta(t)} + \mc B_{\eta(t)} + \mc C_{\eta(t)}).
	    	\end{array}	
	    	\right .
	    \end{equation*}%
	    	   Let us denote with $D^h(\cdot)$ the shrinkage function of $\mc T_{\mathrm{FBF}}^h$. By Lemma \ref{prop:FBF}, $\mc T_{\mathrm{FBF}}^h$ is quasi-shrinking and, therefore, $D^h(\cdot)$ is positive semidefinite. Assumption \ref{as:T_uniformly_quasi_shrinking} is then satisfied with $D(r)=\min_{h\in\{1, ..., H\}} D^h(r) $.	This problem class includes the case when only the selection function ${\phi}$ varies, {i.e., $H=1$.} 
		\end{remark} 
		
		\begin{algorithm}
		\caption{{Optimal v-GNE tracking} via FBF and HSDM}
		\label{alg:Online_fbf}			
	 	\textbf{Initialization.}  Set $x_{i,0} \in \mc X_i$, $\lambda_{i,0} \in \bb R_{\geq 0}^m$, and $\nu_{i,0} \in \bb R^m$, for all $i \in \mc I$. \\
		\textbf{Iteration at time} {$t\in\N_0$} \textbf{of} \textbf{each agent} $i \in \mc I$: 
		\begin{enumerate}
			\item 	 Receives $J_{i,t}(\cdot)$, $g_{i,t}(\cdot)$, {and} $\mc X_{i,t}(\cdot)$. 
			\item	 Assigns $\hat{x}_i^{(1)}\leftarrow  x_{i,t}$, $\hat{\lambda}_i^{(1)}\leftarrow  \lambda_{i, t}$, {and}  $\hat{\nu}_i^{(1)}\leftarrow \nu_{i,t}$.
			\item \textbf{For} $k=1,...,K$: 
		\end{enumerate}
				\begin{enumerate}[(i)]
				\item  Receives $\hat{x}_j^{(k)}$ from agent {$j \in \mc N_i^J$} and $\hat{\lambda}_j^{(k)}, \hat{\nu}_j^{(k)}$ from agent $j \in \mc N_i^\lambda$. 
				\item Updates:				
			{\small 				\begin{align*}
			\tilde x_i^{(k)} \hspace{-1pt} &= \hspace{-1pt} {\prox_{\ell_{i,t} + \iota_{\mc X_{i,t}}}^{\rho_i}}\hspace{-3pt}\left(\hat{x}_i^{(k)} \hspace{-2pt}- \hspace{-2pt} \rho_i (\nabla_{x_i}{f}_{i,t}(\hat{\bm x}^{(k)}) \hspace{-1pt} \right. \\ 
			&\left. \quad + \nabla g_{i,t}(\hat{x}_i^{(k)})^\top \hat{\lambda}_i^{(k)})\hspace{-1pt}\right), \\
			\tilde \lambda_i^{(k)} &= \proj_{\geq 0}\Big(\Big.\hat\lambda_i^{(k)} + \tau_i\Big(\Big. g_{i,t}(\hat x_i^{(k)}) \\ & \quad +\textstyle\sum_{j \in \mc N_i^\lambda} \left(\hat\nu_i^{(k)}- \hat\nu_j^{(k)} - \hat\lambda_i^{(k)} + \hat\lambda_j^{(k)} \right) \Big.\Big) \Big.\Big), \\
			\tilde \nu_i^{(k)} &= \hat\nu_i^{(k)} - \sigma_i\textstyle\sum_{j \in \mc N_i^\lambda} \left(\hat\lambda_i^{(k)}- \hat\lambda_j^{(k)} \right).
		\end{align*}
	}
				\item Receives $\tilde x_j^{(k)}$ from agent $j \in \mc N_i^J$ and $\tilde \lambda_j^{(k)}, \tilde \nu_j^{(k)}$ from agent $j \in \mc N_i^\lambda$. 
				\item Updates: 
				{\small 	\begin{align*}
				\begin{split}
					\accentset{\circ} x_i^{(k)} = & \tilde x_j^{(k)} - \rho_i \Big( \nabla_{x_i}f_{i,t}( \tilde{{\bm{x}}}^{(k)}) -\nabla_{x_i}f_{i,t}(\hat{\bm{x}}^{(k)}) + \\ 
					 & \nabla g_{i,t}(\tilde x_i^{(k)})^\top \tilde{\lambda}_i^{(k)}- \nabla g_{i,t}(\hat{x}_i^{(k)})^\top \hat{\lambda}_i^{(k)} \Big),  \\	
					\accentset{\circ} \lambda_i^{(k)} = & \tilde \lambda_i^{(k)} + \tau_i \Big(\Big. g_{i,t}(\tilde{x}_i^{(k)}) - g_{i,t}(\hat{x}_i^{(k)}) + \\
					& \textstyle\sum_{j \in \mc N_i^\lambda}\Big(\Big. \tilde \nu_i^{(k)} - \hat{\nu}_i^{(k)} - \tilde \nu_j^{(k)} + \hat{\nu}_j^{(k)} \Big.\Big) - \\
					& \textstyle\sum_{j \in \mc N_i^\lambda}\Big(\Big. \tilde \lambda_i^{(k)} - \hat{\lambda}_i^{(k)} - \tilde \lambda_j^{(k)} + \hat{\lambda}_j^{(k)} \Big.\Big) \Big.\Big) ,\\
					\accentset{\circ} \nu_i^{(k)} = & \tilde \nu_i^{(k)} -  \sigma_i  \textstyle\sum_{j \in \mc N_i^\lambda} \left( \tilde \lambda_i^{(k)} - \hat\lambda_i^{(k)} - \tilde \lambda_j^{(k)} + \hat\lambda_j^{(k)} \right) . \notag
					\end{split}
					\end{align*} }
				\item Sends $(\accentset{\circ}{x}_i^{(k)}, \accentset{\circ}{\lambda}_i^{(k)}, \accentset{\circ}{\nu}_i^{(k)})$ to a coordinator and receives $\nabla \phi^t_{\bomega_i}( \accentset{\circ} x_i^{(k)}, \accentset{\circ}\lambda_i^{(k)}, \accentset{\circ}\nu_i^{(k)})$, where $\bomega_i=(x_i, \lambda_i, \nu_i)$.
				\item Updates:
				\begin{align*}	
				\begin{split}
				    &	(\hat{x}_i^{(k+1)}, \hat{\lambda}_i^{(k+1)}, \hat{\nu}_i^{(k+1)}) \\
				    	& =( \accentset{\circ} x_i^{(k)}, \accentset{\circ}\lambda_i^{(k)}, \accentset{\circ}\nu_i^{(k)}) 
					 - \beta \nabla \phi^t_{\bomega_{i}}( \accentset{\circ} x_i^{(k)}, \accentset{\circ}\lambda_i^{(k)}, \accentset{\circ}\nu_i^{(k)}). 
				\end{split}
				\end{align*}
				\end{enumerate}
				\quad\quad \textbf{End For}
				\begin{enumerate}
				\item [4)] Assigns ${x}_{i,t}\leftarrow  \hat{x}_i^{(K+1)}$, ${\lambda}_{i,t}\leftarrow  \hat{\lambda}_i^{(K+1)}$, ${\nu}_{i,t}\leftarrow \hat{\nu}_i^{(K+1)}$.
\end{enumerate}				 

		\end{algorithm}
		
	\end{subsection}

	\begin{remark}
	   The result of this section holds similarly if we substitute the FBF operator with the pFB operator {in \eqref{eq:T_pFB},} which is quasi-shrinking (see the proof of {Theorem} \ref{prop:pFB}){, for cocoercive games {with affine coupling constraints.}} 
	\end{remark}
	\section{Illustrative example}
	\label{sec:illustrative_ex}
	{We consider  a peer-to-peer electricity market clearing problem with operational constraints of the electrical network,  adapted from \cite{belgioioso21}. We assume that each bus of a distribution network consists of one agent that} has access to either a storage unit or a dispatchable generation unit. {Each agent $i \in \mc I$} has decision authority on the power generated $p^{\text{g}}_{i,h}$, the power bought from the main grid $p^{\text{mg}}_{i,h}$, the power drawn from the storage unit $p^{\text{st}}_{i,h}$, the power traded with the trading partners $p^{\mathrm{tr}}_{(i,j),h},~j\in\mc N_i$ and the phase at the bus $\theta_{i,h}$ over the horizon $h=1,...,H$. {Let us} denote $\bm x_{i,h}=\mathrm{col}(p^{\text{g}}_{i,h},p^{\text{mg}}_{i,h}, p^{\text{st}}_{i,h}, 
	\{p^{\mathrm{tr}}_{(i,j),h}\}_{j\in\mc N_i}, \theta_{i,h})$, for all $i\in\mc I$ and $h=1,...,H$,  and denote $\bm x_i := \mathrm{col}(\{\bm x_{i,h}\}_{h=1,...,H}) $, $\boldsymbol{x}:=\mathrm{col}(\{\bm x_i\}_{i\in\mc I})$. Each agent aims at minimizing {its local cost function \cite[Eq. (17)]{belgioioso21}:} 
    \begin{align}
    \begin{split}
    \label{eq:simulations_cost}
	 J_{i}(\bm x) = & \sum_{h=1}^H f^{\text{g}}_{i,h}(p^{\text{g}}_{i,h})+ f^{\text{tr}}_{i,h}(\{p^{\text{tr}}_{(i,j),h}\}_{j\in\mc N_i}) \\
	 &+ f^{\text{mg}}_{i,h}(p^{\text{mg}}_{i,h}, p^{\text{mg}}_{-i,h}),
	 \end{split}
	 \end{align}
	 where  $f^{\text{tr}}_{i,h}$ encodes the cost or revenue of the trading with other agents and  $f^{\text{mg}}_{i,h}$ encodes the cost of purchasing energy from the main grid as in \cite[{Eq. (11)}]{belgioioso21}, while $f^{\text{g}}_{i,h}$ is a linear function which encodes the cost of power generation. The local feasible sets $\mathcal{X}_i, ~i=1,...,N$ include the satisfaction of the power demand at the bus, as well as the operating constraints of the generators and storage units. The shared constraints are of the form
	 $g(\bm x)\leq \boldsymbol{0}_{n_c},$ with $g$  affine. They include  the operating limits of the grid, the trading reciprocity $\{p^{\mathrm{tr}}_{(i,j),h}=-p^{\mathrm{tr}}_{(j,i),h}, ~\forall~ i\in\mc N,~\forall~j\in\mc N_i\}$ and the linearized power flow equations with DC approximation $
	 \{ p^{\text{g}}_{i,h}+p^{\text{st}}_{i,h}+ \iota^{\text{mg}}_{i} \sum_{j\in\mc N} p^{\text{mg}}_{j,h} + \sum_{j\in\mc B_i} B_{ij}(\theta_{i,h}-\theta_{j,h})=0\}$,  where $\iota^{\text{mg}}_{i}\in\{0,1\}$ is $1$ if and only if $i$ is connected to the main grid, {$\mc B_i$ is the set of buses that are connected to bus $i$ on the electric grid} and $B$ is the susceptance matrix. {We note that the game satisfies Assumptions \ref{as:gen_game} and \ref{as:pseudograd}.}   
	  {In addition, we consider the IEEE {13-bus} distribution feeder for our numerical simulations, performed in \textsc{Matlab}.\\
	 We first simulate the  day-ahead market clearing (with 24 hourly time steps) via the standard FBF-based algorithm, which can obtain a v-GNE, and Algorithm \ref{alg:fbf}, which solves the optimal selection problem of this game.
	  Specifically, we consider  the GNE selection function: 
	 \begin{align}
	 \begin{split}
	     & \phi(\bm x) = \sum_{h=1}^H \{  \| \bm p^{\text{g}}_h - \bar{\bm p}^{\text{g}} \|_{Q_{\text{d}}} + \| \bm p^{\text{mg}}_h \|_{Q_{\text{mg}}} + \| \btheta_h - \bar{\btheta} \|_{Q_{\theta}} \\
	     &   +\| G\btheta_h \|_{Q_{\text{pf}}}  +\|  \bm p^{\text{tr}}_h \|_{Q_{\text{tr}}} + \| \bm p^{\text{st}}_h \|_{Q_{\text{st}}} \} + \| \bm\lambda \|_{Q_\lambda} +  \| \bm\nu \|_{Q_\nu},
	     \end{split}
     \label{eq:phi_p2p}
	 \end{align}
	 where we denoted in bold the column stack of the respective variables for each agent and the matrices $Q_{\star}$ are diagonal positive definite. We choose  $\bar{\bm p}^{\text{g}}$ to be the column vector of the maximum generation production for each agent, in order to maximize the renewable energy production, and $\bar{\btheta}$ to be a vector which elements are all equal to the phase of the node connected to the main grid, in order to reduce the grid imbalances. The cost factors {related} to ${\bm p}^{\text{mg}}, {\bm p}^{\text{st}}, {\bm p}^{\text{tr}} $ aim at reducing the burden on the transmission grid, increasing the lifespan of the storage units and reducing the load of the trading platform, respectively. The terms in $\bm\lambda$ and $\bm \nu$ act as regularization of the dual variables. Finally, $G$ is a matrix that maps the phase of the nodes to the power flowing through the lines. In this test, we aim at maximizing the lifespan of the grid lines by setting the non-zero elements of $Q_{\text{pf}}$ to be large. {The solution obtained by Algorithm \ref{alg:fbf} and that of the standard FBF are depicted} in Figure \ref{fig:day_ahead_market}. As expected, {since the v-GNE computed by Algorithm \ref{alg:fbf} minimizes the selection function \eqref{eq:phi_p2p},}  {it has} a lower load on the power lines {than that of the standard FBF. } \begin{figure}
	 	\centering
	 	\includegraphics[width=\columnwidth]{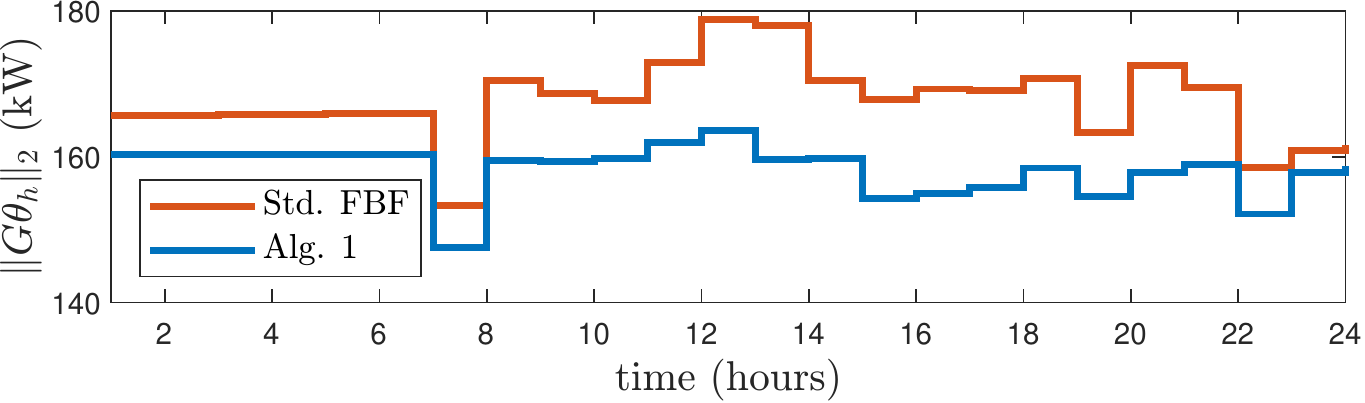}
	 	\caption{Power flowing through the lines (day-ahead market).}
	 	\label{fig:day_ahead_market}
	 \end{figure} \\
	 Secondly, we test Algorithm \ref{alg:Online_fbf} on a real-time market scenario, {formulated as a time-varying game.} The horizon is set to 2 hours, with a sampling time of $15$ minutes. The simulation is run over a 24 hour span for different values, thus resulting in 12 consecutive instances of GNE selection problems. Let us index these problems with $t=1,...,12$.  The cost {function} of each agent is given by \eqref{eq:simulations_cost}, with an additional term $f^{\text{st}}_{i, t}(\{p^{\text{st}}_{i,h}\}_{h=1,...,H})=\|{x}_{\text{ch},t}^i-\bar{x}_{\text{ch},t}^i\|^2_2$. This term penalizes the deviations of the storage units charge state at the end of the $t$-th horizon from the charge state planned in the day-ahead market clearance $\bar{x}_{\text{ch},t}^i$. The charge state at the end of the $t$-th horizon is given by ${x}_{\text{ch},t}^i = {x}_{\text{ch},t-1}^i-\sum_{h=1}^Hp^{\text{st}}_{i,h}$, and the initial state ${x}^i_{\text{ch},0}$ is known. {Because of the variability along the day of the power demand, the local power balance constraint defined in \cite[Eq. (6)]{belgioioso21} depends on $t$. } The {cost} functions and constraints {of the game} are therefore time-varying, with $t$ representing the time index. \ {Furthermore,} in this {scenario,} we aim at computing {a v-GNE that minimizes} the power flowing on the line connecting buses 632 and 671 during peak hours. {Thus, we consider \eqref{eq:phi_p2p} as the selection function at each $t$ where} the element of $Q_{\text{pf}}$  {related} to this line is {time-varying, i.e., it is} set high between {6AM} and {4PM.} { We note that this setup falls into the case considered in Remark \ref{rem:finite_set_operators}, whilst $\{\phi_t\}_{t=1,...,12}$ satisfies Assumption \ref{as:phi_time_varying}.} {We run the simulation for different values of the parameters $K$ and $\beta$ and  Figure \ref{fig:performance_tables} illustrates the results. An increasing $K$ 
	 	results in a diminishing residual, that is, a better convergence to the GNE set, and a diminishing power load on the penalized line during peak hours, {as} expected from the imposed penalty term in the selection function. A diminishing $\beta$ implies a slower reduction of the cost function, which results in a higher cost for small values of $K$, as shown in Figure \ref{fig:table_power}. Figure \ref{fig:real_time_market} depicts some particular trajectories of the power flowing through the penalized line with $\beta=5\cdot 10^{-4}$. }
 
	 	 \begin{figure}
	     \centering
	     \subfigure[Average residual $\|T_{\text{FBF}}(\bomega_t)-\bomega_t\|$.]{
	     \includegraphics[width=0.9\columnwidth]{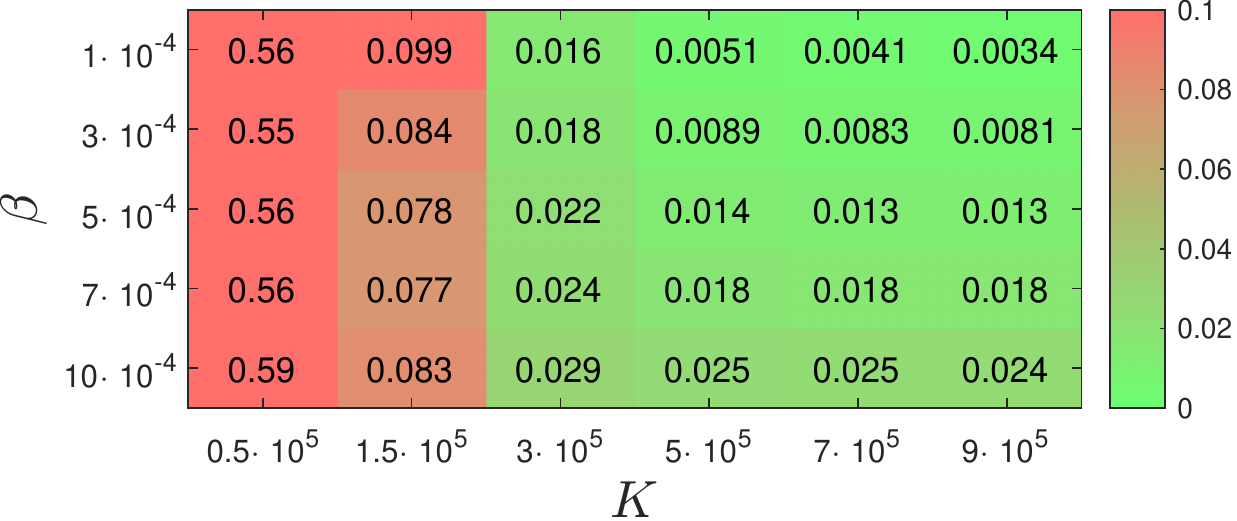}\label{fig:table_residuals}}
	     \subfigure[Average power flow on the line connecting buses 632 and 671 during peak hours (kW).]{
	     \includegraphics[width=0.9\columnwidth]{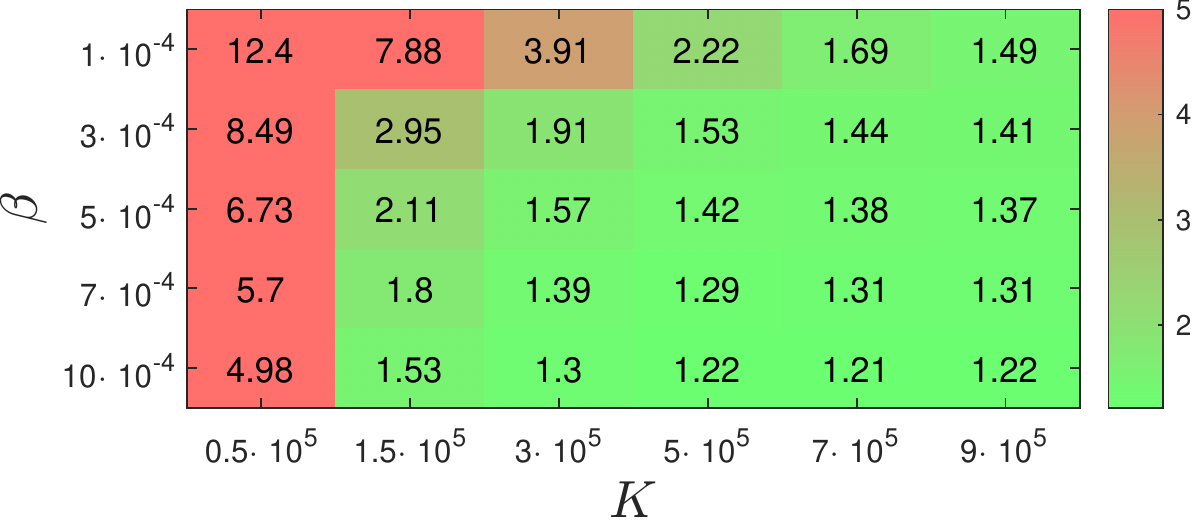}\label{fig:table_power}}
\caption{Algorithm performance for several restarted HSDM parameters. }
\label{fig:performance_tables}
	 \end{figure}
	 	 \begin{figure}
	     \centering
	     \includegraphics[width=\columnwidth]{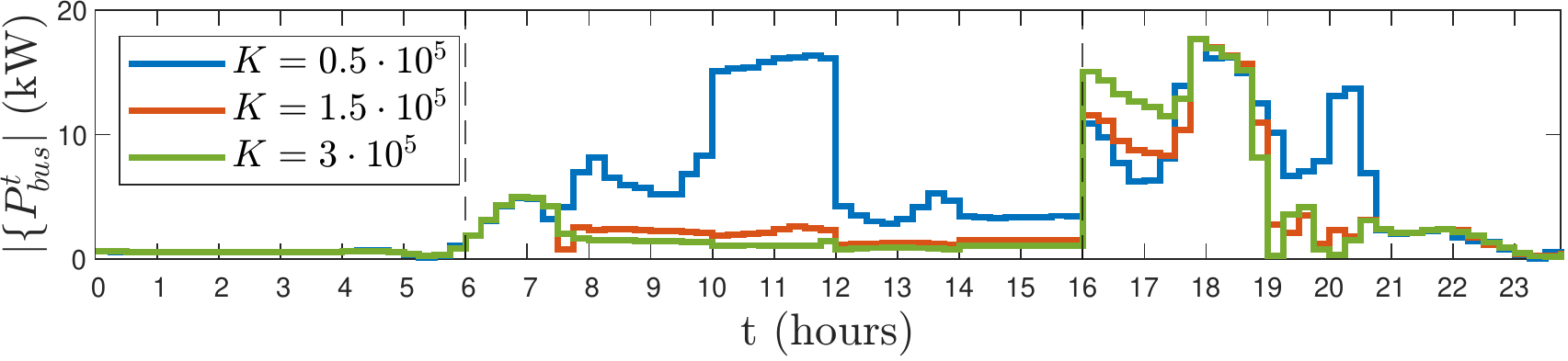}
	     \caption{Power flowing on the line connecting buses 632 and 671 (real-time market) with ${\beta=5\cdot 10^{-4}}.$}
	     \label{fig:real_time_market}
	 \end{figure}
  

	\section{Conclusion}
{The optimal generalized Nash equilibrium selection problem in monotone games {can be solved} {distributively} by combining the hybrid steepest descent method with an appropriate fixed-point operator.  The key requirement to guarantee convergence  to the set of optimal generalized Nash equilibria is the  quasi-shrinking property, which {holds true} for {certain} fixed-point operators. The hybrid steepest descent method can also be modified to track a time-varying optimal generalized Nash equilibria. The resulting approach is suitable for real-time decision {making} {in multi-agent dynamic environments.}}

	
	\appendices
	\section{Proof of Lemma \ref{le:demi_closedness_quasi_shrinking}}
	\label{appendix:proof:le:demi_closedness_quasi_shrinking}
\color{black} Let us proceed by contradiction. We assume that there exists $r>0$ such that $D_{\Psi}(r)=0$. Then, by the definition of $D_{\Psi}(\cdot)$ {in \eqref{eq:D},} there exists a sequence $(\bomega_k)_{k\in\N}\in (\fix(\mc T)_{\geq r}^{\Psi})\bigcap C$ such that 
$$\lim_{k\xrightarrow[]{}\infty} \dist_{\Psi}(\bomega_k, \fix(\mc T)) - \dist_{\Psi}(\mc T (\bomega_k), \fix(\mc T)) =0. $$
By the definition of projection, {we have}
\begin{align}
    \dist_{\Psi}(\mc T (\bomega_k), \fix(\mc T)) &= \| \mc T (\bomega_k) - \proj^{\Psi}_{\fix(\mc T)}(\mc T(\bomega_k)) \|_\Psi 
    \notag\\ & \leq \| \mc T (\bomega_k) - \proj^{\Psi}_{\fix(\mc T)}(\bomega_k) \|_{\Psi}. \label{eq:proof:projection_bound}
\end{align} 
By the quasi-nonexpansiveness of $\mc T$ and the latter inequality,
\begin{align*}
    \begin{split}
        0&\leq  \underbrace{\| \bomega_k - \proj^{\Psi}_{\fix(\mc T)} (\bomega_k)\|_{\Psi}}_{=\dist_{\Psi}(\bomega_k, \fix(\mc T))} -   \| \mc T (\bomega_k) - \proj^{\Psi}_{\fix(\mc T)}(\bomega_k) \|_{\Psi} \\
        &\leq \dist_{\Psi}(\bomega_k, \fix(\mc T)) - \dist_{\Psi}(\mc T (\bomega_k), \fix(\mc T)) \xrightarrow[]{k\xrightarrow{}\infty} 0.
    \end{split}
\end{align*} 
It follows that
$$ \lim_{k\xrightarrow{}\infty} \hspace{-2pt} \| \bomega_k - \proj^{\Psi}_{\fix(\hspace{-1pt} \mc T\hspace{-1pt} )} (\bomega_k)\|_{\Psi}\hspace{-1pt} - \hspace{-1pt}  \| \mc T (\bomega_k) - \proj^{\Psi}_{\fix(\hspace{-1pt} \mc T\hspace{-1pt} )}(\bomega_k) \|_{\Psi}\hspace{-2pt}  = \hspace{-2pt} 0. $$
{By} \eqref{eq:condition_quasi_shrinking}, {we then have that} 
\begin{align*}
    \begin{split}
         & \| \bomega_k  - \mc T_2(\bomega_k) \|^2_{\Psi}\leq  \\ 
         &  \tfrac{1}{\gamma}( \| \bomega_k \hspace{-1pt} - \hspace{-1pt}  \proj^{\Psi}_{\fix(\mc T)}(\bomega_k) \|^2_{\Psi} \hspace{-1pt} -\hspace{-1pt}  \| \mc T(\bomega_k)-\proj^{\Psi}_{\fix(\mc T)}(\bomega_k) \|^2_{\Psi} ) \hspace{-1pt}  \leq  \\
 & \tfrac{2d}{\gamma} ( \| \bomega_k \hspace{-1pt} - \hspace{-1pt}  \proj^{\Psi}_{\fix(\mc T)}(\bomega_k) \|_{\Psi} \hspace{-1pt} - \hspace{-1pt}  \| \mc T(\bomega_k) \hspace{-1pt} - \hspace{-1pt} \proj^{\Psi}_{\fix(\mc T)}(\bomega_k) \|_{\Psi} ),
    \end{split}
\end{align*}
where the latter inequality follows from $a^2-b^2=(a-b)(a+b)$ for $a,b\in\R$ and where we substituted  $d:=\sup_{\bomega\in C} \| \bomega_k - \bomega\|_{\Psi}$, which is finite {since the set} $C$ {is} compact. We conclude that
\begin{equation}
\label{eq:proof:convergence_residual}
    \lim_{k\xrightarrow{}\infty} \| \bomega_k - \mc T_2(\bomega_k) \|^2_{\Psi} =0.
\end{equation}
By the Bolzano-Weierstrass theorem and the boundedness of $\bomega_k$, there exists a convergent subsequence $(\bomega_{k_l})_{l\in\N}$ with accumulation point $\bomega^{\infty}$.  By \eqref{eq:proof:convergence_residual}, $  \lim_{l\xrightarrow{}\infty}  \mc T_2(\bomega_{k_l}) = \bomega^{\infty}.$

By the demiclosedness of ${\Id}- \mc T_2$ and by $\fix(\mc T_2) \subset \fix(\mc T)$, 
$ \bomega^{\infty} - \mc T_2(\bomega^{\infty}) =0 \Rightarrow \bomega^{\infty}\in \fix(\mc T_2) \Rightarrow \bomega^{\infty}\in  \fix(\mc T).$ However, since $ (\fix(\mc T)_{\geq r}^{\Psi})\bigcap C$ is a closed set, then $\bomega^{\infty}\in \fix(\mc T)_{\geq r}^{\Psi}$, which is in contradiction with $\bomega^{\infty}\in  \fix(\mc T)$. 
\qedd 

	\section{Properties of operators $\mc A$, $\mc B$, and $\mc C$ in \eqref{eq:op_A}--\eqref{eq:op_C}}
	\label{ap:prop_ABC}
	\begin{lemma}
		\label{le:max_mon_ABC}
		Let Assumption \ref{as:gen_game} hold. Then, the operators $\mc A$, $\mc B$, and $\mc C$ in \eqref{eq:op_A}--\eqref{eq:op_C} are maximally monotone. Thus, $\mc A+\mc B+\mc C$ is also maximally monotone. \eod 
	\end{lemma}
	\begin{proof}
		  By Assumption \ref{as:gen_game}, $\nc_{{\mc X_i}}$ and  $\partial \ell_i$ are maximally monotone \cite[Thm. 20.25 \& Example 20.26]{bauschke11}. {The operator $\mc A$ is thus maximally monotone by \cite[ Prop. 20.23 \& Cor. 25.5]{bauschke11}. }    {The operator $F$ is maximally monotone}  by Assumption \ref{as:pseudograd} and by continuity in Assumption \ref{as:Lips}. {Meanwhile} $L$ is a linear positive semidefinite operator {and, therefore, it is maximally monotone; thus, the operator $\mc B$ is maximally monotone.} {We can write} $\mc C = \mc C_1 + \mc C_2$, where $\mc C_1 = \col(\{\langle \nabla_{x_i} g_i(x_i), \lambda_i \rangle \}_{i \in \mc I}, -\{g_i(x_i)\}_{i \in \mc I}, \0_{Nm})$ and $\mc C_2 = \col(\0_{n}, -(L \otimes I_m)\bm \nu, (L\otimes I_m)\bm \lambda )$. The operator $\mc C_1$ is maximally monotone by continuity and by noting that, for any $\bomega, \bomega' \in \bb R^{n}\times\bb R^{Nm}_{\geq 0} \times \bb R^{Nm}$,
		\begin{align*}
			& \langle \mc C_1(\bm \omega) - \mc C_1(\bm \omega'), \bm \omega - \bm \omega' \rangle \\
			&=\textstyle\sum_{i \in \mc I} \langle g_i(x_i')- g_i(x_i)-\nabla_{x_i} g_i(x_i)^\top (x_i' - x_i), \lambda_i \rangle\\
			&\quad + \textstyle\sum_{i \in \mc I} \langle g_i(x_i)- g_i(x_i')-\nabla_{x_i} g_i(x_i')^\top (x_i - x_i'), \lambda_i' \rangle \geq 0,%
		\end{align*}
		where the inequality follows by the convexity of $g_i$. 
		As $\mc C_2$ is a linear skew-symmetric operator, it is maximally monotone \cite[Example 20.35]{bauschke11}. By invoking \cite[Cor. 25.5]{bauschke11}, {the result follows.}
	\end{proof}
	
	\begin{lemma}
		\label{le:lips_BC}
		Let Assumptions \ref{as:gen_game} and \ref{as:Lips} hold. Then the operators $\mc B$, $\mc C$,and $\mc B + \mc C$, defined in \eqref{eq:op_B}--\eqref{eq:op_C},  are Lipschitz continuous. 
	\end{lemma}

	\begin{proof}
		Due to Assumption \ref{as:Lips}, the operator $\mc B$ is $L_F$-Lipschitz continuous. Lipschitz continuity of $\mc C$ can be evaluated as follows. Similarly {to} the proof of Lemma \ref{le:max_mon_ABC}, let us split $\mc C = \mc C_1 + \mc C_2$. {The operator $\mc C_2$ is Lipschitz continuous by linearity, while Lipschitz continuity of $\mc C_1$ is shown  as follows.} Let us denote the bound of $\nabla_{x_i} g_i(x_i)$ by $b_{\nabla g_i}$, i.e., $\|\nabla_{x_i} g_i(x_i)\| \leq b_{\nabla g_i}$ (c.f. Assumption \ref{as:Lips})  and the bound of $\lambda_i$ by $b_{\lambda_i}$, for all $i \in \mc I$, which exists due to \cite[Prop. 3.3]{auslender00}. For any $\bm \omega, \bm \omega' \in \bb R^{n+2Nm}$,
		{\small 
		\begin{align*}
			&\| \mc C_1(\bm \omega) - \mc C_1(\bm \omega') \|^2 \\
			&\overset{\{1\}}{\leq} \textstyle\sum_{i \in \mc I}\Big(\Big.  2 \| \nabla_{x_i} g_i(x_i)^\top (\lambda_i -\lambda_i')\|^2 + \|g_i(x_i)-g_i(x_i') \|^2 \\
			&\qquad + 2\| (\nabla_{x_i} g_i(x_i)- \nabla_{x_i} g_i(x_i'))^\top \lambda_i'\|^2\Big.\Big)\\
			&\overset{\{2\}}{\leq} \textstyle\sum_{i \in \mc I}\Big(\Big. 2 \| \nabla_{x_i} g_i(x_i)^\top \|^2 \|\lambda_i -\lambda_i' \|^2  + b_{\nabla g_i}^2 \| x_i - x_i'\|^2  \\
			&\qquad + 2\|\lambda_i' \|^2 \| \nabla_{x_i} g_i(x_i)- \nabla_{x_i} g_i(x_i')\|^2\Big.\Big)\\\
			&\overset{\{3\}}{\leq} \hspace{-1pt}  \textstyle\sum_{i \in \mc I} \hspace{-1pt} \Big(2 b_{\nabla g_i}^2 \|\lambda_i \hspace{-1pt} - \hspace{-1pt} \lambda_i' \|^2  \hspace{-1pt} + \hspace{-1pt} (2b_{\lambda_i}^2 L_{g}^2 + b_{\nabla g_i}^2) \| x_i \hspace{-1pt} - \hspace{-1pt} x_i'\|^2  \\
			&\leq \textstyle\sum_{i \in \mc I} {\max(2b_{\nabla g_i}^2, 2b_{\lambda_i}^2 L_{g}^2 + b_{\nabla g_i}^2)} \| \bm \omega_i - \bm \omega_i'\|^2,
		\end{align*}}%
		{where $\{1\}$ follows by adding and subtracting the term $ \nabla_{x_i} g_i(x_i)^\top \lambda_i'$ and by the bound $\|a+b\|^2\leq 2\|a\|^2 + 2\|b\|^2$; $\{2\}$ is obtained by the Cauchy-Schwartz inequality and by the fact that $g_i$ is Lipschitz since it has a bounded gradient; $\{3\}$ is obtained by the  Lipschitz continuity of $\nabla_{x_i} g_i$. Hence, $\mc C_1$ is $L_{\mc C_1}$-Lipschitz continuous, where $L_{\mc C_1} = \max_{i \in \mc I} (\max(2b_{\nabla g_i}, \sqrt{2b_{\lambda_i}^2 L_{g}^2 + b_{\nabla g_i}^2)})$.
{Since} the sum of Lipschitz continuous operators is Lipschitz continuous, the result follows.}
	\end{proof}

	\begin{lemma}
		Let $(\bm x^\star,\bm \lambda^\star)$ be a solution to the monotone inclusion in \eqref{eq:mon_incl}. Then, $(\bm x^\star,\bm \lambda^\star)$ is also a solution to the monotone inclusion in \eqref{eq:KKT}. \eod 
	\end{lemma}
	\begin{proof}
		The proof follows that of \cite[Thm. 2(i)]{yi19}. 
	\end{proof}
	
	\section{{Results and Proofs of Section \ref{sec:op_gne}}}
	\label{ap:le:fix_FBF}
	
	{The following lemma shows the equivalence between $\zer(\mc A+\mc B+\mc C)$ and $\fix(\mc T_{\mathrm{FBF}})$.}
	
		\begin{lemma}
			\label{le:fix_FBF} 
	Let Assumptions \ref{as:gen_game}, \ref{as:pseudograd}, \ref{as:Lips}, and \ref{as:stepsizeFBF} hold. Furthermore, let $\mc T_{\mathrm{FBF}}$ be defined by \eqref{eq:T_FBF} while $\mc A$, $\mc B$, and $\mc C$ be defined in \eqref{eq:op_A}--\eqref{eq:op_C}. Then, $\fix(\mc T_{\mathrm{FBF}}) = \zer(\mc A + \mc B + \mc C)$. \eod 
\end{lemma}
\begin{proof}
{The proof is analogous to {that} of \cite[Prop. 1]{franci20}.}
\end{proof}

{The following lemma is used to prove {the} quasi-shrinking property of the FBF operator \eqref{eq:T_FBF}.}
{
\begin{lemma}
\label{lem:demiclosed_FBF}
    Let $\mc A$ and $\mc B$  maximally monotone and $\mc B$ continuous. Let
    $$ \mc T={(\Id +\Psi^{-1}\mc A)^{-1}}(\Id-\Psi^{-1}\mc B). $$ 
    Then $\Id-\mc T$ is demiclosed at $0$. \eod 
\end{lemma}
\begin{proof}
Let us consider a sequence $(v_k)_{k\in\N}$ such that
$$ \lim_{k\xrightarrow{}\infty} v_k = v , \qquad \lim_{k\xrightarrow{}\infty} \mc (\Id-\mc T) (v_k) =0.$$
We want to prove that $v- \mc T (v) = 0 $ or, equivalently, $v\in \fix(\mc T)$. Let us define 
$ u_k:=\mc (\Id-\mc T) (v_k). $ 
Then, 
\begin{align*}
    \begin{split}
        &v_k-u_k = (\Id +\Psi^{-1}\mc A)^{-1}(\Id-  {\Psi^{-1}} \mc B) (v_k) \\
        &\Leftrightarrow (\Id- {\Psi^{-1}} \mc B) (v_k)  \in (\Id+{\Psi^{-1}}\mc A)( v_k-u_k)  \\
        & \Leftrightarrow v_k - {\Psi^{-1}}\mc B (v_k) +u_k - v_k \in {\Psi^{-1}}\mc A( v_k-u_k)  \\
        & \Leftrightarrow  \mc - \mc B (v_k) + \Psi u_k   \in  \mc A( v_k-u_k). 
    \end{split}
\end{align*}

By the continuity of $\mc B$ and \cite[Fact 1.19]{bauschke11}, we conclude that  $ \lim_{k\xrightarrow{}\infty}  - {\mc B (v_k)} + \Psi u_k = -\mc B(v).  $
By \cite[Prop. 20.37]{bauschke11},  {$\gph(\mc A)$} is closed. Therefore, since $\lim_{k\xrightarrow{}\infty} v_k-u_k = v$, we conclude that 
$ -\mc B( v) \in  \mc A( v) . $ 
By \cite[Prop. 26.1(iv)]{bauschke11}, we obtain $v\in \fix(\mc T)$.
\end{proof}

	\subsection{Proof of Lemma \ref{le:FBF_q_nonexp}}
	\label{pf:le:FBF_q_nonexp}
		{By Lemmas \ref{le:max_mon_ABC} and \ref{le:lips_BC}, the operator $\mc A$ is maximally monotone whereas the operator $\mc B + \mc C$ is maximally monotone and Lipschitz continuous with Lipschitz constant denoted by $L_B$.} Then, \cite[Cor. 1]{franci20} shows that $\mc T_{\mathrm{FBF}}$ is quasi-nonexpansive when the step size matrix $\Psi$, satisfy Assumption \ref{as:stepsizeFBF}. Specifically, it holds that \cite[Prop. 2]{franci20}:
		\begin{equation}
			\|\mc T_{\mathrm{FBF}} (\bm \omega) - \bm \omega^\star  \|_{\Psi}^2 \leq \|\bm \omega - \bm \omega^\star  \|_{\Psi}^2 - (L_B/\mu_{\min}(\Psi))^2 \|\tilde{\bm \omega} - \bm \omega \|_{\Psi}^2,
			\label{eq:q_nonexp}
		\end{equation}
		where $\bm \omega^\star \in \fix(\mc T_{\mathrm{FBF}})$, $\mu_{\min}(\Psi)$ is the smallest eigenvalue of $\Psi$ and $\tilde{\bm \omega} =(\Id +\Psi^{-1}\mc A)^{-1}(\Id- \Psi^{-1}(\mc B + \mc C))(\bm \omega)$. { Finally, we prove that {$\mc T_{\mathrm{FBF}}$} is quasi-shrinking by invoking Lemma \ref{le:demi_closedness_quasi_shrinking}. Specifically, we choose $\mc T_2=(\Id +\Psi^{-1}\mc A)^{-1}(\Id- \Psi^{-1}(\mc B + \mc C))$. As shown in the proof of {\cite[Prop. 1]{franci20},} $\fix(\mc T_2)=\zer(\mc A+\mc B + \mc C) = \fix(\mc T_{\mathrm{FBF}})$. Moreover, Lemma \ref{lem:demiclosed_FBF} shows that $\Id - \mc T_2$ is demiclosed at 0 and \eqref{eq:q_nonexp} is indeed the inequality in \eqref{eq:condition_quasi_shrinking} for $\mc T_{\mathrm{FBF}}$. }
		\qedd
		\begin{remark}
			Although \cite[Cor. 1]{franci20} shows quasi-nonexpansiveness of $\mc T_{\mathrm{FBF}}$ and \cite[Prop. 2]{franci20} shows the inequality in \eqref{eq:q_nonexp} for Problem \eqref{eq:gen_game} with a linear coupling constraint, these results also holds for nonlinear functions $g_i(x_i)$, for all $i \in \mc I$, as long as Assumption \ref{as:Lips} holds, since the operator $\mc C$ in \eqref{eq:op_C} remains Lipschitz continuous.
		\end{remark}

	\subsection{Proof of Lemma \ref{le:bounded_HSDM}}
	\label{pf:le:bounded_HSDM}
		Firstly, we show that, for an arbitrary $\bm{\omega}^\star \in \fix(\mc T_{\mathrm{FBF}})$, 
		\begin{align}
			\|\mc T_{\mathrm{FBF}} (\bm \omega) - \bm \omega^\star  \|_{\Psi}^2 < \|\bm \omega - \bm \omega^\star  \|_{\Psi}^2, \label{eq:attr_q_nonexp}
		\end{align}
		for all $\bm \omega \notin  \fix(\mc T_{\mathrm{FBF}})$. To this end, let us recall the inequality \eqref{eq:q_nonexp} in the proof of Lemma \ref{le:FBF_q_nonexp}:
		\begin{equation}
			\|\mc T_{\mathrm{FBF}} (\bm \omega) - \bm \omega^\star  \|_{\Psi}^2 \leq \|\bm \omega - \bm \omega^\star  \|_{\Psi}^2 - (L_B/\mu_{\min}(\Psi))^2 \|\tilde{\bm \omega} - \bm \omega \|_{\Psi}^2, \notag
		\end{equation}
		which holds for any $\bm \omega^\star \in \fix(\mc T_{\mathrm{FBF}})$ and $\bm \omega \in \dom(\mc T_{\mathrm{FBF}})$. Furthermore, we consider any $\bm \omega \notin \fix({\mc T_{\mathrm{FBF}}})$. 
		Since $\tilde{\bm \omega} = (\Id +{\Psi^{-1}\mc A})^{-1}(\Id - \Psi^{-1}(\mc B + \mc C)))(\bm \omega)$, when $\tilde{\bm \omega} = \bm \omega$, it holds that
		\begin{align}
			\Psi (\tilde{\bm \omega} - \bm \omega)-(\mc B + \mc C)(\bm \omega) + (\mc B + \mc C)(\tilde{\bm \omega}) &\in (\mc A+\mc B + \mc C)(\tilde {\bm \omega}) \notag \\
			\Leftrightarrow 0 &\in (\mc A+\mc B + \mc C)(\tilde {\bm \omega}), \notag
		\end{align}
		implying that $\bm \omega = \tilde{\bm \omega} \in \fix(\mc T_{\mathrm{FBF}})$. Hence, $\tilde{\bm \omega} \neq \bm \omega$ if $\bm \omega \notin \fix({\mc T_{\mathrm{FBF}}})$. We observe from the preceding inequality that when $\tilde{\bm \omega} \neq \bm \omega$, the inequality \eqref{eq:attr_q_nonexp} holds. 
		
		By using the inequality \eqref{eq:attr_q_nonexp} and the fact that $\fix(\mc T_{\mathrm{FBF}})$ is bounded, we can then show that for any arbitrary fixed point $\bm \omega^\star \in \fix(\mc T_{\mathrm{FBF}})$, there exists $R>0$ satisfying $\inf_{\|\bm \omega -\bm \omega^\star\|\geq R} (\|\bm \omega - \bm \omega^\star \| - \|\mc T_{\mathrm{FBF}}(\bm \omega) -\mc T_{\mathrm{FBF}}(\bm \omega^\star)\|) > 0$. The proof of the {previous} statement {is analogous to} that of \cite[Lem. 1]{ogura03}, which claims a similar inequality for attracting non-expansive operators. Finally, we follow the proof of \cite[Thm. 2]{ogura03}, which claims the boundedness of the HSDM sequence with an attracting non-expansive operator $\mc T$, since, for some $R>0$, the inequality $\inf_{\|\bm \omega -\bm \omega^\star\|\geq R} (\|\bm \omega - \bm \omega^\star \| - \|\mc T_{\mathrm{FBF}}(\bm \omega) -\mc T_{\mathrm{FBF}}(\bm \omega^\star)\|) > 0$ holds not only for attracting non-expansive operators but also $\mc T_{\mathrm{FBF}}$, $\nabla \phi$ is monotone and Lipschitz continuous (Assumption \ref{as:phi}), and the step size $\beta^{(k)}$ is non-summable but square summable (Assumption {\ref{as:beta}).} \qedd
		
		\subsection{Proof of Theorem \ref{prop:FBF}}
		\label{pf:prop:FBF}
			
			Let $\tilde{\bm \omega}^{(k)} = (\tilde{\bm x}^{(k)},\tilde{\bm \lambda}^{(k)},\tilde{\bm \nu}^{(k)})$ and $\accentset{\circ}{\bm \omega}^{(k)} = (\accentset{\circ}{\bm x}^{(k)},\accentset{\circ}{\bm \lambda}^{(k)}, \accentset{\circ}{\bm \nu}^{(k)} ) $, where $\tilde{\bm x}^{(k)} = \col(\{\tilde x_i\}_{i \in \mc I})$ and the other variables are defined similarly. The updates of $\tilde{\bm \omega}^{(k)}$ in Step 2 of Algorithm \ref{alg:fbf} can be compactly written as 
			$$ \tilde{\bm \omega}^{(k)} = (\Id +{\Psi^{-1}\mc A})^{-1}(\Id -\Psi^{-1}(\mc B + \mc C))(\bm \omega^{(k)}),$$
			whereas the updates of $\accentset{\circ}{\bm \omega}^{(k)} $ in Step 4 of Algorithm \ref{alg:fbf} can be compactly written as 
			$\accentset{\circ}{\bm \omega}^{(k)} =  \tilde{\bm \omega}^{(k)} - \Psi^{-1}(\mc B + \mc C)( \tilde{\bm \omega}^{(k)} - \bm \omega^{(k)}),$
			implying that $\accentset{\circ}{\bm \omega}^{(k)} = \mc T_{\mathrm{FBF}}(\bm \omega^{(k)})$ and the updates in \eqref{eq:hsdm_fbf} is compactly written as 
			\begin{equation}
				\bm \omega^{(k+1)} = \mc T_{\mathrm{FBF}}(\bm \omega^{(k)}) - \beta^{(k)} \nabla \phi(\mc T_{\mathrm{FBF}}(\bm \omega^{(k)})),
				\label{eq:fbf_hsdm}%
			\end{equation}
			which is the {HSDM} applied to $\mc T_{\mathrm{FBF}}$.

			We can then invoke {Lemma} \ref{th:hsdm_cvx}
			to claim the hypothesis. By {Lemma \ref{le:fix_FBF},} $\fix(\mc T_{\mathrm{FBF}})=\zer(\mc A+\mc B + \mc C)$; {therefore $\fix(\mc T_{\mathrm{FBF}})$}  is non-empty and bounded. Moreover, by Assumption {\ref{as:beta},} the step size $\beta^{(k)}$ meets the conditions in Lemma \ref{th:hsdm_cvx}. 
			Lemma \ref{le:FBF_q_nonexp} shows that $\mc T_{\mathrm{FBF}}$ is quasi-nonexpansive and  quasi-shrinking on any bounded closed convex set, $C$ such that $C \cap \fix(\mc T_{\mathrm{FBF}}) \neq \varnothing$. On the other hand, Lemma \ref{le:bounded_HSDM} shows that the FBF-HSDM  sequence $(\bm \omega^{(k)})_{k \in \bb N}$ obtained by the iterations in \eqref{eq:fbf_hsdm} is bounded, i.e., for any $\bm \omega^\star \in \fix(\mc T_{\mathrm{FBF}})$, there exists a positive {finite} $R(\bm \omega^\star)$ such that $\|\bm \omega^{(k)} - \bm \omega^\star \| \leq R(\bm \omega^\star)$. 
			Therefore, for an arbitrarily chosen $\bm \omega^\star \in \fix(\mc T_{\mathrm{FBF}})$, 
			we can construct the following bounded closed set
			$	\mathfrak{B}(\bm \omega^\star) := \{ x \in \dom(\mc T_{\mathrm{FBF}}) \mid \|x - \bm \omega^\star \| \leq R(\bm \omega^\star)  \},$ 
			on which the sequence $(\bm \omega^{(k)})_{k \in \bb N}$ lies. Moreover,
			we can observe that indeed $\mathfrak{B} \cap \fix({\mc T_{\mathrm{FBF}}} )\neq \varnothing$, since $\bm \omega^\star \in \mathfrak{B}$ is a fixed point of $\mc T_{\mathrm{FBF}}$. Hence, $\mc T_{\mathrm{FBF}}$ is quasi-shrinking on $\mathfrak{B}$, which completes the proof. \qedd 
	
	\section{{Proofs of Section \ref{sec:special_cases}}}
	
	\subsection{Proof of Theorem \ref{prop:pFB}}
	\label{pf:prop:pFB}
		First, we observe that in Algorithm \ref{alg:pFB}, $\accentset{\circ}{\bm \omega}^{(k)}=(\accentset{\circ}{\bm x}^{(k)},\accentset{\circ}{\bm \lambda}^{(k)},\accentset{\circ}{\bm \nu}^{(k)} )$ is updated by using $\mc T_{\mathrm{pFB}}$ in \eqref{eq:T_pFB}, i.e.,  $\accentset{\circ}{\bm \omega}^{(k)}= \mc T_{\mathrm{pFB}}(\bm \omega^{(k)})$  \cite[Section 4, Algorithm 1]{yi19}. Hence, we can see that $\bm \omega^{(k)}$ is updated via the HSDM method, i.e., {
			\begin{equation}
				\bm \omega^{(k+1)} = \mc T_{\mathrm{pFB}}(\bm \omega^{(k)}) - \beta^{(k)} \nabla \phi(\mc T_{\mathrm{pFB}}(\bm \omega^{(k)})), \label{eq:pfb_hsdm}
		\end{equation}}%
		Similarly {to} the proof of {Theorem} \ref{prop:FBF}, due to the boundedness of $\fix(\mc T_{\mathrm{pFB}})=\zer(\mc A + \mc B+\mc C) \neq \varnothing$ and the step size rule of $\beta^{(k)}$ in Assumption {\ref{as:beta},} we can invoke Lemma \ref{th:hsdm_cvx}.
		{
		{Specifically,} the operator $\mc T_{\mathrm{pFB}}$ is averaged nonexpansive when Assumptions \ref{as:gen_game}, \ref{as:pseudograd}, \ref{as:Lips}, and \ref{as:pseudograd_cocoercive}--\ref{as:stepsizepFB} hold \cite[Thm. 3]{yi19}. Therefore, $\mc T_{\mathrm{pFB}}$ is also quasi-nonexpansive \cite[Section 4.1]{bauschke11}. By \cite[Prop. 4.35 (iii)]{bauschke11}, the condition in \eqref{eq:condition_quasi_shrinking} holds with $\mc T_2 = \mc T$. By \cite[Thm. 4.27]{bauschke11}, $\Id - \mc T_{\mathrm{pFB}}$ is demiclosed at 0. Therefore, by Lemma \ref{le:demi_closedness_quasi_shrinking}, $\mc T_{\mathrm{pFB}}$ is quasi-shrinking on any closed bounded convex set whose intersection with $\fix(\mc T_{\mathrm{pFB}})$ is nonempty. Furthermore, since $\mc T_{\mathrm{pFB}}$ is averaged nonexpansive, $\mc T_{\mathrm{pFB}}$ is attracting. Therefore, by \cite[Thm. 2]{ogura03} and due to the choice of the step size $\beta^{(k)}$ in Assumption \ref{as:beta}, the sequence generated by \eqref{eq:pfb_hsdm} is bounded.} \qedd 
	\section{{Proofs of Section \ref{sec:tv_GNEselect}}} 
\subsection{Preliminary results}	 {First, we show a series of {preliminary} results in Lemmas \ref{lemma:bound_of_sequence}--\ref{le:auxiliary_lemma_online_bounds} that lead to the proofs of Lemma \ref{lemma:finite_hsdm_bound_single_iteration} and {Theorem} \ref{prop:asymptotic_limit_online_hsdm}.}  The proofs of this section are provided in the standard Euclidean norm for ease of notation. However, the case for any $\Psi$-induced norm, with $\Psi\succ 0$, follows verbatim.} 
	{First, Lemma \ref{lemma:bound_of_sequence} shows} the convergence of a particular sequence and can be regarded as a finite-iteration version of \cite[Lem. 1]{yamada05}. 
		\begin{lemma}
			\label{lemma:bound_of_sequence}
			Let $\psi:\R_{\geq 0}\to \R_{\geq 0}$ be non-decreasing and non-negative. Let a sequence $(b^{(k)})_{k\in\N}$ be non-increasing, non-negative. Let $(a^{(k)})_{k\in\N} \subset [0,\infty)$ satisfy 
			\begin{equation}
				\label{eq:key_sequence}
				a^{(k+1)}\leq a^{(k)}-\psi(a^{(k)})+b^{(k+1)}.
			\end{equation} 
			Let $K\in\N$. If there exists $\xi>0$ such that $\psi(\xi)\geq \max\{2b^{(1)}, \frac{2}{K-1}a^{(1)}\}$, then 
			\begin{equation}
				\label{eq:bound_sequence}
				a^{(k)}\leq \xi+b^{(k)}, ~~~~ \forall~ k\geq K.
			\end{equation}
		\end{lemma}
		\begin{proof}
			Let us first show that there exists an $M\in \N$, $M\leq K$ such that $a^{(M)}\leq\xi$. 
			We proceed by contradiction, assuming that $a^{(k)}>\xi$ $\forall k=1, ..., K$. Then, by noting that $\psi(\cdot)$ is non-decreasing and that  $\psi(\xi)\geq 2b^{(k)}$ for all $k\in\N$, we have
			\begin{align*}
				\begin{split}
					a^{(k+1)} & \leq a^{(k)}-\psi(a^{(k)})+b^{(k+1)} \\ & \leq a^{(k)}-\psi(\xi)+\tfrac{1}{2}\psi(\xi)=a^{(k)}-\tfrac{1}{2}\psi(\xi).
				\end{split}
			\end{align*}  
			By iterating the latter relation and recalling that $ \psi(\xi)\geq \frac{2}{K-1}a^{(1)} $, we find {that}
			\begin{align*}
			\begin{split}
			    & a^{(k+1)}\leq a^{(1)}-\tfrac{k}{2}\psi(\xi) \leq a^{(1)} - \tfrac{k}{K-1}a^{(1)}.
		    \end{split}
			\end{align*}
			For $k=K$, we then obtain the contradiction $a^{(K+1)}<0$. {Thus,} there exists $M\leq K$ such that $a^{(M)}\leq\xi$. We then proceed by induction to prove (\ref{eq:bound_sequence}). Let us prove that, if $a^{(k)}\leq \xi + b^{(k)}$ then $a^{(k+1)}\leq \xi + b^{(k+1)}$ for all $k\geq M$. We distinguish two cases:\\
1) Case $a^{(k)}<\xi$. Then, by (\ref{eq:key_sequence}) and by the non-negativity of $\psi(\cdot)$,
				$a^{(k+1)} \hspace{-2pt} \leq \hspace{-2pt} a^{(k)}+b^{(k+1)}  \hspace{-2pt} < \hspace{-2pt} \xi+b^{(k+1)}.$ \\
2) Case $\xi\leq a^{(k)}\leq \xi+b^{(k)}$. Then, by the non-decreasing property of $\psi$, $a^{(k)}\geq \xi \Rightarrow \psi(a^{(k)})\geq\psi(\xi)$. By the assumptions, $\psi(\xi)\geq 2b^{(1)}$ and by the non-incresing property of $(b^k)_{k\in\N}$,  $2b^{(1)}\geq b^{(k)} + b^{(k+1)}$. We thus obtain $\psi(a^{(k)})\geq b^{(k)} + b^{(k+1)}$. Substituting into (\ref{eq:key_sequence}) leads to
						\begin{align*}
							a^{(k+1)} &\leq a^{(k)}-\psi(a^{(k)})+b^{(k+1)} \\
							&\leq a^{(k)}-b^{(k)} - b^{(k+1)} +b^{(k+1)} = a^{(k)} -b^{(k)}\leq \xi. 
						\end{align*}
				We conclude by induction that $a^{(k)} \leq \xi +b^{(k)}$ for all $k\geq M$ and, since $M \leq K$, the claim in (\ref{eq:bound_sequence}) immediately follows.
		\end{proof}
		\begin{lemma}
			\label{lemma:useful_bounds}
			Let $\mc T$ quasi-nonexpansive and $\mc F$ strongly monotone, such that $\| \mc F(\bomega) \|\leq U $ for all $\bomega \in \mathrm{im}(\mc T)$. Let $(\bomega^{(k)})_{k\in\N}$ be generated from (\ref{eq:hsdm}) with constant stepsize $\beta^{(k)} =\beta >0$ for all $k$. Let $K\in\N$  and let $\omegaopt$ be the solution of $\VI(\mc F, \fix(\mc T))$. If there exists $\xi$ such that the shrinkage function $D(\cdot)$ of $\mc T$, {defined in \eqref{eq:D},} satisfies $D(\xi) \geq \max \{ 2\beta U, 2\frac{\dist(\omega^{(1)},\fix(\mc T))}{K-1} \}$, then the following {inequalities} hold:
				\begin{align}   \label{eq:lemma:useful_bounds_claim_A} \sup_{k\geq K} \dist(\bm \omega^{(k)}, \fix(\mc T)) &\leq \xi + \beta U, \\
					 \label{eq:lemma:useful_bounds_claim_B} \sup_{k\geq K} \|\mc T(\bm \omega^{(k)})-\bm \omega^{(k)} \| &\leq 2(\xi + \beta U), \\
				\label{eq:lemma:useful_bounds_claim_C}  \sup_{k\geq K}  \langle \mc T(\bm \omega^{(k)}) - \bm \omegaopt, - \mc F(\bm \omegaopt) \rangle & \leq  3(\xi + \beta U) \| \mc F(\bm \omegaopt)\|.
				\end{align}
		\end{lemma}
		\noindent\emph{Proof.} (i.) For all $k$, it holds by the definition of distance and by the algorithm definition in (\ref{eq:hsdm}) that: 
				\begin{align}
						&\dist(\bm \omega^{(k+1)}, \hspace{-1pt} \fix(\mc T)) \hspace{-2pt} \leq \hspace{-2pt} \| \bm \omega^{(k+1)} \hspace{-2pt} - \hspace{-2pt} \proj_{\fix(\mc T)}(\mc T(\bm \omega^{(k)}))\| \hspace{-2pt} = \notag \notag\\
						& \| \mc T(\bm \omega^{(k)}) - \beta \mc F( \mc T(\bm \omega^{(k)})) - \proj_{\fix(\mc T)}(\mc T(\bm \omega^{(k)}))\|  \leq  \notag \notag\\
						&\underbrace{\| \mc T(\bm \omega^{(k)}) - \proj_{\fix(\mc T)}(\mc T(\bm \omega^{(k)})) \|}_{=\dist(\mc T(\bm \omega^{(k)}), \fix(\mc T))} + \beta \|\mc F ( \mc T(\bm \omega^{(k)}))\| \leq \notag\\ & \dist(\mc T(\bm \omega^{(k)}), \fix(\mc T)) + \beta U. \label{eq:iterative_bound_distance}
				\end{align}
				Let us define $a^{(k)}:=\dist(\bm \omega^{(k)}, \fix(\mc T))$. Then, from (\ref{eq:iterative_bound_distance}) we find immediately $ a^{(k+1)} - \beta U \leq \dist(\mc T(\bm \omega^{(k)}), \fix(\mc T)) $. By the definition of shrinkage function in (\ref{eq:D}) and the latter inequality, we can write
				\begin{align*}
					\begin{split}
						& D(a^{(k)}) \leq a^{(k)} - \dist(\mc T(\bm \omega^{(k)}), \fix(\mc T)) \leq \\ &  a^{(k)} - a^{(k+1)} + \beta U  
						 \Rightarrow a^{(k+1)} {\leq} a^{(k)} + \beta U -  D(a^{(k)}),
					\end{split}
				\end{align*}
				which defines a sequence of the kind in (\ref{eq:key_sequence}) with $\psi(\cdot)=D(\cdot)$ and $b^{(k)}= \beta U$ for all $k$.
				By Lemma \ref{lemma:bound_of_sequence}, then $\dist(\bm \omega^{(k)}, \fix(\mc T))\leq \xi + \beta U$ for all $k\geq K$.  \smallskip \\ 
				(ii.) {By the triangle inequality,} we can write $\|\mc T(\bm \omega^{(k)})-\bm \omega^{(k)}\|  \leq   {\| \mc T(\bm \omega^{(k)})-\proj_{\fix(\mc T)} (\bm \omega^{(k)}) \|} + \| \proj_{\fix(\mc T)}(\bm \omega^{(k)}) - \bm \omega^{(k)} \|.$	By quasi-nonexpansiveness of $\mc T$, we obtain, for all $k\geq K$,
				\begin{align*}
					\begin{split} 
						& 
						{\| \mc T(\bm \omega^{(k)})-\proj_{\fix(\mc T)} (\bm \omega^{(k)}) \|}  \leq   \\
						& 
						\| \bm \omega^{(k)} - \proj_{\fix(\mc T)}(\bm \omega^{(k)}) \| = \dist(\bm \omega^{(k)}, \fix(\mc T))  \\
						& 
						\Rightarrow \|\mc T(\bm \omega^{(k)})-\bm \omega^{(k)}\| \leq 2\dist(\bm \omega^{(k)}, \fix(\mc T)).
					\end{split}
				\end{align*}
				Finally, {combining the last inequality and \eqref{eq:lemma:useful_bounds_claim_A} yields \eqref{eq:lemma:useful_bounds_claim_B}.}\\
				(iii) By the Cauchy-Schwarz inequality, we can write
				\begin{align}
						& \langle \mc  T(\bm \omega^{(k)})-\bm \omegaopt, -\mc F(\bm \omegaopt) \rangle = \notag \\
						&\langle \mc T(\bm \omega^{(k)}) - \bm \omega^{(k)}, - \mc F(\bm \omegaopt) \rangle + \langle \bm \omega^{(k)} - \bm \omegaopt, - \mc F(\bm \omegaopt) \rangle \leq \notag\\
						&   \| \mc T(\bm \omega^{(k)}) - \bm \omega^{(k)} \| \| \mc F(\bm \omegaopt) \| + \langle \bm \omega^{(k)} - \bm \omegaopt, - \mc F(\bm \omegaopt) \rangle. \label{eq:lemma:useful_bounds_claim_iii}
				\end{align}
				{Based on} 	 (\ref{eq:lemma:useful_bounds_claim_B}), for all $k\geq K$, we {can} bound the first term on {the right-hand side of \eqref{eq:lemma:useful_bounds_claim_iii}} by
						$ 
						\| \mc T(\bm \omega^{(k)}) - \bm \omega^{(k)} \| \| \mc F(\bm \omegaopt) \| \leq 2(\xi + \beta U) \| \mc F(\bm \omegaopt) \|
					$ 
				{and rewrite the second term as}
				\begin{align*}
					\begin{split}
						\langle \bm \omega^{(k)} - \bm \omegaopt, - \mc F(\bm \omegaopt) \rangle= &
						  {\langle \bm \omega^{(k)}  \hspace{-2pt} - \hspace{-2pt} \proj_{\fix(\mc T)}(\bm \omega^{(k)}), \hspace{-1pt} -\mc F(\bm \omegaopt) \rangle} \\
						+ & 
						{\langle \proj_{\fix(\mc T)}(\bm \omega^{(k)})-\bm \omegaopt, -\mc F(\bm \omegaopt) \rangle}.
					\end{split}
				\end{align*}
				{We observe that the second addend is non-positive by the definition of $\VI$ solution.} {By} applying the Cauchy-Schwarz inequality, the definition of projection, and (\ref{eq:lemma:useful_bounds_claim_A}), we obtain 
		\begin{align*}
					\begin{split}
						& \langle \mc  T(\bm \omega^{(k)})-\bm \omegaopt, -\mc F(\bm \omegaopt) \rangle \leq \\
						&  2(\xi   +  \beta U) \| \mc F(\bm \omegaopt) \| \hspace{-1pt} + \hspace{-1pt} \| \bm \omega^{(k)} \hspace{-1pt} - \hspace{-1pt} \proj_{\fix(\mc T)}(\bm \omega^{(k)}) \| \| \mc F(\bm \omegaopt)\| =  \\
						&  2(\xi +  \lambda  U) \| \mc F(\bm \omegaopt) \| \hspace{-1pt}+ \hspace{-1pt} \dist(\bm \omega^{(k)}, \fix(\mc T))  \| {\mc F}(\bm \omegaopt)\|\leq \\
						& 3(\xi + \lambda U)  \| \mc  F(\bm \omegaopt) \|.
					\end{split}
		\end{align*}\qedd 
		{\begin{lemma}
		\label{le:auxiliary_lemma_online_bounds}
		Let Assumptions \ref{as:phi_time_varying}--{\ref{as:bounded_gradient}} hold. For any $t\in\N$, let $\bomega_{t+1}$ be generated from the step at time $t$ of the restarted HSDM algorithm in \eqref{eq:algorithm_online_HSDM}. Let $D_t(\cdot)$ be the shrinkage function of $\mc T_t$ as defined in \eqref{eq:D}. If {there exists} $\xi>0$ {such that}
			\begin{equation} 
				\label{eq:condition_shrinkage_function}
				D_t(\xi) \geq \max{\left\{2\beta U, 2\tfrac{\dist(\bm \omega_{t},\fix(\mc T_t))}{K-1} \right\} },
			\end{equation}
            then, the bound in \eqref{eq:contractive-like_bound}, i.e.,
            \begin{equation*}
            	\|\bm \omega_{t+1}-\bm \omegaopt_t\|^2 \leq \left(1-{\tau(\beta)}\right)^K \| \bm \omega_{t}-\bm \omegaopt_t \|^2 + \gamma,
            \end{equation*}
             holds with
			\begin{equation} 
				\label{eq:definition_gamma}
				\gamma = \tfrac{\beta}{\tau(\beta)}U  (6\xi+11\beta U).
			\end{equation}
		\end{lemma} }
		\begin{proof}{
			{Let us define} the operator 
		$  \mc T^{\beta}_t(\bm \omega) := \mc T_t(\bm \omega) - \beta\nabla \phi_t (\mc T_t (\bm \omega)). $ 
			By $\mc T_t(\bm \omegaopt_t)=\bm \omegaopt_t$ and by the definition of the algorithm {in \eqref{eq:algorithm_online_HSDM},}
					$\|\bm \bomega_{t+1}- \bm \omegaopt_t\|^2 = \| \mc T^{\beta}_t(\bm y^{(K)})-\mc T_t(\bm \omegaopt_t) \|^2 $. 
			We sum and subtract $\beta\nabla \phi_t(\bm \omegaopt_t)$ and substitute  $\mc T^{\beta}_t$ to obtain
			\begin{align*}
			    \begin{split}
			    &	\|\bm \bomega_{t+1}- \bm \omegaopt_t\|^2\\
			    & =  \| \mc T^{\beta }_t(\bm y^{(K)})-\mc T_t(\bm \omegaopt_t) +\beta\nabla \phi_t(\bm \omegaopt_t) - \beta\nabla \phi_t(\bm \omegaopt_t)\|^2  \\
				& = \| \mc T^{\beta }_t(\bm y^{(K)})-\mc T^{\beta }_t(\bm \omegaopt_t)  - \beta\nabla \phi_t(\bm \omegaopt_t)\|^2.
				\end{split}
			\end{align*}
			Expanding the square ${\{1\}}$, expanding $\mc T^{\beta }_t$ $\{2\}$, and regrouping $\{3\}$ leads to 
			\begin{align}
			    \begin{split}
			    	&	\|\bm \bomega_{t+1}- \bm \omegaopt_t\|^2\\
					&\overset{\{1\}}{=} \| \mc T^{\beta }_t(\bm y^{(K)}) - \mc T^{\beta }_t(\bm \omegaopt_t) \|^2 + \beta ^2\|\nabla\phi_t(\bm \omegaopt_t)\|^2 \\
					&+ 2 \langle \mc T^{\beta }_t(\bm y^{(K)})-\mc T^{\beta }_t(\bm \omegaopt_t), -\beta\nabla \phi_t(\bm \omegaopt_t) \rangle   \\
					& \overset{\{2\}}{=} \| \mc T^{\beta }_t(\bm y^{(K)}) - \mc T^{\beta }_t(\bm \omegaopt_t) \|^2 + \beta ^2\| \nabla \phi_t (\bm \omegaopt_t)\|^2 - \\
					& 2 \beta\langle \mc T_t(\bm y^{(K)})  - \hspace{-1pt} \beta\nabla \phi_t( \mc T_t(\bm y^{(K)})) - \\
					&\mc T_t(\bm \omegaopt_t)\hspace{-1pt}+\hspace{-1pt} \beta\nabla \phi_t( \mc T_t(\bm \omegaopt_t)), \nabla \phi_t(\bm \omegaopt_t) \rangle  \\ 
					& \overset{\{3\}}{=} \| \mc T^{\beta }_t(\bm y^{(K)}) - \mc T^{\beta }_t(\bm \omegaopt_t) \|^2+ \beta ^2\| \nabla\phi_t(\bm \omegaopt_t)\|^2  \\
					& + 2 \beta\langle \mc T_t(\bm y^{(K)})-\bm \omegaopt_t, - \nabla \phi_t(\bm \omegaopt_t) \rangle  \\
					& + 2\beta ^2\langle \nabla\phi_t (\mc T_t(\bm y^{(K)}))  - \nabla\phi_t(\bm \omegaopt_t), \nabla\phi_t(\bm \omegaopt_t) \rangle. 
				\end{split}
				\label{eq:first_bounds_geometric_seq}
			\end{align}
			We note that, by applying the Cauchy-Schwarz, the triangle inequalities and Assumption \ref{as:bounded_gradient}, we have 
					$\langle \nabla\phi_t ( \mc T_t(\bm y^{(K)}))  - \nabla\phi_t(\bm \omegaopt_t), \nabla\phi_t(\bm \omegaopt_t) \rangle 
					\leq \| \nabla\phi_t ( \mc T_t(\bm y^{(K)}))  - \nabla\phi_t(\bm \omegaopt_t)\| \|\nabla\phi_t(\bm \omegaopt_t)\| 
					\leq (U+\|\nabla\phi_t(\bm \omegaopt_t)\|)\|\nabla\phi_t(\bm \omegaopt_t)\|. 
					$ 
			By \eqref{eq:condition_shrinkage_function}, the bounds in Lemma \ref{lemma:useful_bounds} hold. We then substitute in (\ref{eq:first_bounds_geometric_seq}), the latter relation, and the bound in \eqref{eq:lemma:useful_bounds_claim_C} to obtain
			\begin{align*}
				\begin{split}
					&\|\bm \bomega_{t+1}- \bm \omegaopt_t\|^2  \leq \| \mc T^{\beta }_t(\bm y^{(K)}) - \mc T^{\beta }_t(\bm \omegaopt_t) \|^2  +  \\
					& 6\beta (\xi \hspace{-2pt} + \hspace{-2pt} \beta U) \| \nabla\phi_t(\bm \omegaopt_t)\| \hspace{-2pt} + \hspace{-2pt} \beta ^2 (2U+3\|\nabla\phi_t(\bm \omegaopt_t)\|)\|\nabla\phi_t(\bm \omegaopt_t)\|.
				\end{split}
			\end{align*}
			Applying Assumption \ref{as:bounded_gradient} and rearranging the terms leads to
			\begin{align}
					& \|\bm \bomega_{t+1}- \bm \omegaopt_t\|^2  \leq \| \mc T^{\beta }_t(\bm y^{(K)}) - \mc T^{\beta }_t(\bm \omegaopt_t) \|^2 +  \notag\\
					& \quad\quad 6\beta (\xi + \beta U) \| \nabla\phi_t(\bm \omegaopt_t)\| + \beta ^2 5U\|\nabla\phi_t(\bm \omegaopt_t)\|  \notag \\
					& \leq \| \mc T^{\beta }_t(\bm y^{(K)}) - \mc T^{\beta }_t(\bm \omegaopt_t) \|^2 + \beta (6\xi + 11\beta U)U \notag \\
					& \leq \| \mc T^{\beta  }_t(\bm y^{(K)}) - \mc T^{\beta  }_t(\bm \omegaopt_t) \|^2 + {\tau(\beta  )} \gamma .	\label{eq:proof:main_bound:step_2}
			\end{align}
			By {quasi-nonexpansiveness} of $\mc T_t$  {as well as} strong monotonicity and Lipschitz continuity of $ \nabla \phi_t $, we can apply \cite[Lem. 4a]{yamada05} to obtain
			$ \| \mc T^{\beta  }_t(\bomega)- \mc T^{\beta  }_t(\bar{\bomega}) \| \leq (1-\tau(\beta)) \| \bomega - \bar{\bomega} \|, $
			for all $\bomega\in \mathrm{dom}(\mc T^{\beta}_t), \bar{\bomega}\in\fix(\mc T_t)$, which we substitute in (\ref{eq:proof:main_bound:step_2}) to obtain
			\begin{align*}
				\begin{split}
					& \|\bm \bomega_{t+1}- \bm \omegaopt_t\|^2  \leq \left(1-{\tau(\beta)}\right)^2 \| \bm y^{(K)} - \bm \omegaopt_t\|^2 + { \tau(\beta)} \gamma\\
					& \leq \left(1-{\tau(\beta)}\right) \| \bm y^{(K)} - \bm \omegaopt_t\|^2 + { \tau(\beta)}\gamma .
				\end{split}
			\end{align*}
			{By iterating, we obtain} 
			\begin{align*}
			    \begin{split}
					& \|\bm \bomega_{t+1}- \bm \omegaopt_t\|^2 \leq
					 \\
					&  \left(1-{\tau(\beta)}\right)^2 \| \bm y^{(K-1)} - \bm \omegaopt_t\|^2   +  \left(1-{\tau(\beta)}\right) { \tau(\beta)}\gamma + { \tau(\beta)}\gamma\\ 
					&\leq \dots
					\leq \hspace{-2pt} (1-{\tau(\beta)})^K \| \bm y^{(1)} - \bm \omegaopt_t\|^2 \hspace{-2pt} + \hspace{-2pt}  \sum_{j=0}^{K-1} \hspace{-2pt} \left(1-{\tau(\beta)}\right)^j { \tau(\beta)}\gamma  \\ 
        			& \leq \left(1-{\tau(\beta)}\right)^K \| \bm y^{(1)} - \bm \omegaopt_t\|^2 + \sum_{j=0}^{\infty}  \left(1-{\tau(\beta)}\right)^j { \tau(\beta)}\gamma.
				\end{split}
			\end{align*}
				Applying the geometric series convergence and recalling from \eqref{eq:algorithm_online_HSDM}  
				that $\bm y^{(1)}=\bomega_{t} $ leads to \eqref{eq:contractive-like_bound}. } 
		\end{proof}

		 \subsection{Proof of Lemma \ref{lemma:finite_hsdm_bound_single_iteration}}
		 \label{appendix:proof:lemma:finite_hsdm_bound_single_iteration}
		{ Let us consider 
   $ \xi:=  \frac{\gamma \sigma}{12 U}.$ 
Since $\mc T_t$ is quasi-shrinking, the shrinkage function $D_t$ of $\mc T_t$ satisfies $D_t(\xi)>0$. Thus, there exist $\bar\beta\in(0, \frac{2\sigma}{L_\phi^2})$ and $K$ such that, for any $\beta\in (0, \bar\beta]$,

\begin{equation} 
	D_t(\xi) \geq \max{\left\{2\beta U, 2\tfrac{\dist(\bm \omega_{t},\fix(\mc T_t))}{K-1} \right\} }.
\end{equation} 
\begin{remark}
    As $D_t(\xi)$ decreases with $\gamma$, for smaller values of $\gamma$ a smaller stepsize $\beta$ and a larger $K$ are necessary. 
\end{remark}

It can be verified that $\lim_{\beta\rightarrow0^+} \tfrac{\beta}{\tau(\beta)}= \tfrac{1}{\sigma}$. Then, 
\begin{align} 
\label{eq:proof:limit_expression}
\begin{split}
\lim_{\beta\rightarrow 0^+} \tfrac{\beta}{\tau(\beta)}(6\xi + 11\beta U)U& = \tfrac{6\xi U}{\sigma} \leq \tfrac{1}{2}\gamma,  
\end{split}
\end{align}
We thus find $\beta\in(0,\bar{\beta}]$ small enough, such that 
\begin{equation} 
\label{eq:proof:choice_beta}
\tfrac{\beta}{\tau(\beta)}(6\xi + 11\beta  U) U \leq \gamma. 
\end{equation}
Hence, the hypothesis holds by invoking Lemma \ref{le:auxiliary_lemma_online_bounds}.}
		\qedd
	    \subsection{Proof of Theorem \ref{prop:asymptotic_limit_online_hsdm}}
	    \label{appendix:proof:prop:asymptotic_limit_online_hsdm}
	{
	    We begin the proof by constructing a suitable stepsize $\bar\beta$ and number of iterations $\bar{K}$. We then proceed with proving that the statement holds for the chosen variables. 
	    Let us first define the auxiliary variable
           $ \xi=  \frac{\gamma \sigma}{12 U}.$ 
	   By \eqref{eq:proof:limit_expression}, we can choose a small enough $\bar\beta \in(0,\min\{\frac{2\sigma}{L^2_{\varphi}}, \frac{D(\xi)}{2U}\}) $, such that
	    \begin{equation} 
	    \label{eq:prop:condition_gamma_xi}
	    \tfrac{\bar\beta}{\tau(\bar\beta)}(6\xi + 11\bar\beta  U) U \leq \gamma . 
	    \end{equation}
	    We now define 
	    $\alpha(K):=(1-\tau(\bar\beta))^{K}.$ 
	    Since $\tau(\bar\beta)\in (0,1)$, $\alpha$ is decreasing with $K$. We can then choose $K_1$, such that $ \alpha(K_1)<\frac{1}{2}$.
				Then, we define the mapping ${a}:\N_{\geq K_1}\to\R$
        	    \begin{equation}
        			\label{eq:prop:asymptotic_limit_online_hsdm:condition_a_bar}
        			{a}(K) =  \max \left\{ \|\bomega_1\| + \sup_{\bomega\in\mc Y} \|\bomega\|, \sqrt{  \tfrac{2\alpha(K) \delta_1^2 + \gamma }{1-2\alpha(K)}} \right \}, 
        		\end{equation} 
				{We can verify} that ${a(\cdot)}$ is non-increasing. Consequently,  the sequence 
				$ \left(\frac{2({a}(K) +\delta_2)}{{K}-1}\right)_{K\geq K_1} $
				is decreasing. We can then choose any {sufficiently large} $ \bar{K}\geq K_1 $, such that
				\begin{equation}
				   \label{eq:prop:condition_rho_a_bar}
				   D(\xi)\geq  \tfrac{2(\bar{a} +\delta_2)}{\bar{K}-1}, 
				\end{equation} 
				where $\bar{a}:=a(\bar{K})$. {We also define $\bar{\alpha}:=\alpha(\bar{K}) $}.
			
				We now prove {by induction} that 
				\begin{equation}
					\|\bm \omega_t-\bm \omega_{t-1}^{\star}\| \leq \bar{a} ~~~~\text{for all} ~t>1.
				\end{equation}
				To that end, we first show that
				\begin{equation}
					\label{eq:prop:asymptotic_limit_online_hsdm:first_step_induction}
					\|\bm \omega_t- \bm \omega_{t-1}^{\star}\|\leq \bar{a} \Rightarrow \| \bm \omega_{t+1}-\bm \omega_t^{\star} \|\leq \bar{a}.
				\end{equation}
				Let us then write
				\begin{align}
       				\dist(\bm \omega_t, & \fix (\mc T_t))  \overset{\{1\}}{\leq}  \| \bomega_t-\proj_{\fix{(\mc T_t)}}(\omegaopt_{t-1}) \| \notag\\
       				& \overset{\{2\}}{\leq} \| \bomega_t - \omegaopt_{t-1} \| + \| \omegaopt_{t-1} - \proj_{\fix{(\mc T_t)}}(\omegaopt_{t-1}) \|  \notag\\
				     &  \overset{\{3\}}{\leq} \| \bomega_t - \omegaopt_{t-1} \| +\delta_2  \leq \bar{a} + \delta_2, 	\label{eq:prop:steps_bound_distance}
				\end{align} 
				where $\{1\}$ follows from the definition of distance, $\{2\}$ from the triangle inequality and $\{3\}$ from Assumption \ref{ass:var_lim}.\ref{ass:GNE_set_bounded_variations}. Then, by Assumption \ref{as:T_uniformly_quasi_shrinking}, by the choice $\bar\beta\leq \frac{D(\xi)}{2U} $ and \eqref{eq:prop:condition_rho_a_bar},
				\begin{align}
				\label{eq:proof:condition_D}
				\begin{split}
				D_t({\xi}) \geq \max\left\{2\bar\beta U, \tfrac{2(\bar{a}+\delta_2)}{\bar{K}-1} \right \}  \geq \max\left\{2\bar{\beta} U, \tfrac{2\dist(\bm \omega_t, \fix(\mc T_t))}{\bar{K} -1} \right\}.  
				\end{split}
				\end{align} 
			     By Lemma \ref{le:auxiliary_lemma_online_bounds} and \eqref{eq:prop:condition_gamma_xi}, we then have 
				\begin{equation}
				\label{eq:proof:contractivity_property}
				    \| \bm \omega_{t+1}-\bm \omega_t^{\star} \|^2 \leq \bar\alpha \| \bm \omega_t-\bm \omega_t^{\star}  \|^2 + \gamma .
				\end{equation} 
				Applying on \eqref{eq:proof:contractivity_property} the triangle inequality, the fact $(a+b)^2 \leq 2a^2 + 2b^2$ and Assumption \ref{ass:var_lim}.\ref{ass:variability_limited} leads to 
				\begin{align} 
				    \| \bm \omega_{t+1}-\bm \omegaopt_t\|^2 & \leq 2\alpha (\| \bm \omega_t-\bm \omegaopt_{t-1}\|^2 + \| \omegaopt_{t-1} - \omegaopt_t\|^2) + \gamma  \notag\\
				    &\leq 2\bar\alpha (\| \bm \omega_t-\bm \omegaopt_{t-1}\|^2 + \delta_1^2) + \gamma  \notag\\
				    & \leq 2\bar\alpha (\bar{a}^2 + \delta_1^2) + \gamma .
					\label{eq:prop:asymptotic_limit_online_hsdm:steps_bound}
				\end{align}
				Finally, {by   \eqref{eq:prop:asymptotic_limit_online_hsdm:condition_a_bar}, it holds that}
				\begin{equation}
				\label{eq:proof:condition_a_bar_implication}
				    2\bar\alpha (\bar{a}^2 + \delta_1^2) + \gamma \leq \bar{a}^2 \Leftrightarrow  \bar{a}^2 \geq \tfrac{2\bar\alpha \delta_1^2 + \gamma }{1-2\bar\alpha}.
				\end{equation} 
				Thus, we obtain
				$ \| \bm \omega_{t+1}-\bm \omegaopt_t \|^2 \leq \bar{a}^2. $ We now continue the induction argument by proving 
				\begin{equation}
					\label{eq:prop:asymptotic_limit_online_hsdm:second_step_induction}
					\|\bm \omega_2-\bm \omegaopt_1\|^2 \leq \bar{a}^2.
				\end{equation}
				From  the triangle inequality and from \eqref{eq:prop:asymptotic_limit_online_hsdm:condition_a_bar},  
				$\| \bm \omega_1- \bm \omegaopt_1 \| \leq \| \bm \omega_1 \| + \| \bm \omegaopt_1 \| \leq \bar{a}. $ 
				From the definition of distance, we  obtain
				\begin{equation}
				\label{eq:proof:distance_omega_1}
				    \dist(\bm \omega_1, \fix(\mc T_1)) \leq \| \bm \omega_1- \bm \omegaopt_1 \| \leq \bar{a}\leq \bar{a} + \delta_2.
				\end{equation} 
				Then,
					$D_t({\xi})\geq D({\xi}) 
					\geq \max\left\{2\beta U, \frac{2(\bar{a}+\delta_2)}{\bar{K} -1} \right \} 
					\geq \max\left\{2\beta U, \frac{2\dist(\bm \omega_1, \fix(\mc T_1))}{\bar{K} -1} \right\}. $ 
 By Lemma \ref{le:auxiliary_lemma_online_bounds} and \eqref{eq:prop:condition_gamma_xi}, we find 				\begin{align*}
					\begin{split}
						&\| \bm \omega_2-\bm \omega_1^{\star} \|^2 \leq  \bar\alpha \| \bm \omega_1-\bm \omega_1^{\star} \|^2  + \gamma  . 
					\end{split}
				\end{align*}
				By {using \eqref{eq:proof:distance_omega_1} and \eqref{eq:proof:condition_a_bar_implication} to upperbound the right hand side of the last inequality,} we then obtain
				\begin{equation*}
				    \| \bm \omega_2-\bm \omega_1^{\star} \|^2	\leq \bar\alpha \bar{a}^2 + \gamma  \leq \bar\alpha(2\bar{a}^2 + 2\delta_1^2)+ \gamma  \leq \bar{a}^2.
				\end{equation*}
				Therefore, combining (\ref{eq:prop:asymptotic_limit_online_hsdm:first_step_induction}) and (\ref{eq:prop:asymptotic_limit_online_hsdm:second_step_induction}) leads to
				$ \sup_{t>1}  \| \bm \omega_t-\bm \omegaopt_{t-1} \| \leq \bar{a}. $ 
				Recalling that, from Assumption \ref{as:bounded_optimizer_sequence}, $ \omegaopt_{t}\in \mc Y$ for all $t$, this immediately implies 
				$ \dist(\bomega_t, \mc Y) \leq \bar a ~~~ \text{for all}~t>1,  $
				which proves that the sequence is bounded. \\
				We now proceed with proving (\ref{eq:asymptotic_tracking_error}). We note that the relation in \eqref{eq:proof:contractivity_property} holds for all $t$. We then observe that, by the triangle inequality, by $(a+b)^2 \leq 2a + 2b$, and by Assumption \ref{ass:var_lim},
				\begin{align*}
				    \begin{split}
				        \|  \bomega_{t+1}-\omegaopt_{t+1}\|^2 & \leq 2\|  \bomega_{t+1}-\omegaopt_{t}\|^2 + 2\| \omegaopt_{t+1}-\omegaopt_{t}\|^2 \\ 
				        & \leq 2\|  \bomega_{t+1}-\omegaopt_{t}\|^2 + 2\delta_1^2 .
				    \end{split}
				\end{align*}
				{By using} (\ref{eq:proof:contractivity_property}) to {upper bound $\|  \bomega_{t+1}-\omegaopt_{t}\|^2$} and iterating, we find:
				\begin{align*}
					\begin{split}
						& \|  \bomega_{t+1}-\omegaopt_{t+1}\|^2 \leq 2\bar\alpha \|  \bomega_{t}-\omegaopt_{t}\|^2 + 2( \gamma+ \delta_1^2) \\
						& \leq (2\bar\alpha)^2 \|  \bomega_{t-1}-\omegaopt_{t-1}\|^2 + 2 (\gamma + \delta_1^2) + 2\bar\alpha(2 \gamma+ 2\delta_1^2) \\ 
						& \leq \dots \leq (2\bar\alpha)^{t}\| \bomega_1-\omegaopt_1\|^2 + \textstyle\sum_{j=0}^{t-1} (2\bar\alpha)^j (2 \gamma+2\delta_1^2).
					\end{split}
				\end{align*}
				By taking the limit for $t\rightarrow\infty$ and by applying the convergence of the geometric sequence, we obtain (\ref{eq:asymptotic_tracking_error}). 
			\qedd
	
	\subsection{Proof of Corollary \ref{cor:online_FBF}}
	\label{app:pf:cor:online_FBF}
	Steps i--vi of Algorithm \ref{alg:Online_fbf} are analogous to Steps 1--6 of Algorithm \ref{alg:fbf}. Analogously to the proof of {Theorem} \ref{prop:FBF}, we see that the variable $\bm y^{(k)}:=(\hat{x}_i^{(k)}, \hat{\lambda}_i^{(k)}, \hat{\nu}_i^{(k)})$ is updated at each {time step} by $K$ iterations of the HSDM:
	$$\bm y^{(k+1)}= \mc T_{\mathrm{FBF},t}(\bm y^{(k)}) - \beta \nabla \phi_t(\mc T_{\mathrm{FBF},t}(\bm y^{(k)})), ~~ k=1,...,K. $$
	Then, the variable $\bomega_{t}$ is updated as $ \bomega_{t}= \bm y^{(K+1)}.$
	Thus, we see that Algorithm \ref{alg:Online_fbf} is a particular instance of the restarted HSDM algorithm {\eqref{eq:algorithm_online_HSDM}.} By {Theorem} $\ref{prop:asymptotic_limit_online_hsdm}$, $\bomega_t$ is bounded, therefore there exists {a compact set} $\mc Z$  such that $(\bomega_t)\in \mc Z $ for all $t$. By Lemma \ref{le:FBF_q_nonexp}, $\mc T_{\mathrm{FBF},t}$ is quasi-nonexpansive and quasi-shrinking on any bounded, closed convex set 
	{$ C$} such that $C\cap\fix(\mathcal T_{\mathrm{FBF},t})\neq\emptyset$.  In particular, it is quasi-shrinking on any convex set $C\supset \mc Y \cup \mc Z$, {where $\mc Y$ is a compact set such that $\bomega_{t}^\star \in \mc Y, \forall t \in \bb N$ (Assumption \ref{as:bounded_optimizer_sequence}).}  We then find Assumption \ref{as:quasi_nonexp_always} to hold and, by {Theorem} \ref{prop:asymptotic_limit_online_hsdm}, the tracking error is given by \eqref{eq:asymptotic_tracking_error}. \qedd

	\bibliography{bibliography_bibtex}

\begin{thebibliography}{10}

\bibitem{shilov21b}
I.~Shilov, H.~Le~Cadre, and A.~Busic, ``Privacy impact on generalized {Nash}
  equilibrium in peer-to-peer electricity market,'' {\em Operations Research
  Letters}, vol.~49, no.~5, pp.~759--766, 2021.

\bibitem{belgioioso21}
G.~Belgioioso, W.~Ananduta, S.~Grammatico, and C.~Ocampo-Martinez,
  ``Operationally-safe peer-to-peer energy trading in distribution grids: A
  game-theoretic market-clearing mechanism,'' {\em IEEE Transactions on Smart
  Grid}, 2022.
\newblock accepted.

\bibitem{Bakhshayesh21}
B.~G. Bakhshayesh and H.~Kebriaei, ``Decentralized equilibrium seeking of joint
  routing and destination planning of electric vehicles: A constrained
  aggregative game approach,'' {\em IEEE Transactions on Intelligent
  Transportation Systems}, pp.~1--10, 2021.
\newblock Early access at https://doi.org/10.1109/TITS.2021.3123207.

\bibitem{pang2008}
J.-S. Pang, G.~Scutari, F.~Facchinei, and C.~Wang, ``Distributed power
  allocation with rate constraints in {Gaussian} parallel interference
  channels,'' {\em IEEE Transactions on Information Theory}, vol.~54, no.~8,
  pp.~3471--3489, 2008.

\bibitem{Scutari2012}
G.~Scutari, D.~P. Palomar, F.~Facchinei, and J.-S. Pang, ``Monotone games for
  cognitive radio systems,'' in {\em Distributed Decision Making and Control},
  pp.~83--112, Springer London, 2012.

\bibitem{Wang14}
J.~Wang, M.~Peng, S.~Jin, and C.~Zhao, ``A generalized {Nash} equilibrium
  approach for robust cognitive radio networks via generalized variational
  inequalities,'' {\em IEEE Transactions on Wireless Communications}, vol.~13,
  no.~7, pp.~3701--3714, 2014.

\bibitem{facchinei10}
F.~Facchinei and C.~Kanzow, ``Generalized {Nash} equilibrium problems,'' {\em
  Annals of Operations Research}, vol.~175, no.~1, pp.~177--211, 2010.

\bibitem{facchinei07a}
F.~Facchinei, A.~Fischer, and V.~Piccialli, ``On generalized {Nash} games and
  variational inequalities,'' {\em Operations Research Letters}, vol.~35,
  no.~2, pp.~159--164, 2007.

\bibitem{kulkarni12}
A.~A. Kulkarni and U.~V. Shanbhag, ``On the variational equilibrium as a
  refinement of the generalized {Nash} equilibrium,'' {\em Automatica},
  vol.~48, no.~1, pp.~45--55, 2012.

\bibitem{paccagnan18}
D.~Paccagnan, B.~Gentile, F.~Parise, M.~Kamgarpour, and J.~Lygeros, ``Nash and
  {Wardrop} equilibria in aggregative games with coupling constraints,'' {\em
  IEEE Transactions on Automatic Control}, vol.~64, no.~4, pp.~1373--1388,
  2018.

\bibitem{belgioioso18}
G.~Belgioioso and S.~Grammatico, ``Projected-gradient algorithms for
  generalized equilibrium seeking in aggregative games are preconditioned
  forward-backward methods,'' in {\em Proceedings of the 2018 European Control
  Conference (ECC)}, pp.~2188--2193, 2018.

\bibitem{yin11}
H.~Yin, U.~V. Shanbhag, and P.~G. Mehta, ``{Nash} equilibrium problems with
  scaled congestion costs and shared constraints,'' {\em IEEE Transactions on
  Automatic Control}, vol.~56, no.~7, pp.~1702--1708, 2011.

\bibitem{belgioioso17}
G.~Belgioioso and S.~Grammatico, ``Semi-decentralized {Nash} equilibrium
  seeking in aggregative games with separable coupling constraints and
  non-differentiable cost functions,'' {\em IEEE Control Systems Letters},
  vol.~1, no.~2, pp.~400--405, 2017.

\bibitem{belgioioso20}
G.~Belgioioso and S.~Grammatico, ``Semi-decentralized generalized {Nash}
  equilibrium seeking in monotone aggregative games,'' {\em IEEE Transactions
  on Automatic Control}, 2021.
\newblock Early access at https://doi.org/10.1109/TAC.2021.3135360.

\bibitem{yi19}
P.~Yi and L.~Pavel, ``An operator splitting approach for distributed
  generalized {Nash} equilibria computation,'' {\em Automatica}, vol.~102,
  pp.~111--121, 2019.

\bibitem{yi19b}
P.~Yi and L.~Pavel, ``Distributed generalized {Nash} equilibria computation of
  monotone games via double-layer preconditioned proximal-point algorithms,''
  {\em IEEE Transactions on Control of Network Systems}, vol.~6, no.~1,
  pp.~299--311, 2019.

\bibitem{franci20}
B.~Franci, M.~Staudigl, and S.~Grammatico, ``Distributed forward-backward
  (half) forward algorithms for generalized {Nash} equilibrium seeking,'' in
  {\em Proceedings of the 2020 European Control Conference (ECC)},
  pp.~1274--1279, IEEE, 2020.

\bibitem{belgioioso20distributed}
G.~Belgioioso, A.~Nedi{\'c}, and S.~Grammatico, ``Distributed generalized
  {Nash} equilibrium seeking in aggregative games on time-varying networks,''
  {\em IEEE Transactions on Automatic Control}, vol.~66, no.~5, pp.~2061--2075,
  2020.

\bibitem{bianchi22}
M.~Bianchi, G.~Belgioioso, and S.~Grammatico, ``Fast generalized {Nash}
  equilibrium seeking under partial-decision information,'' {\em Automatica},
  vol.~136, no.~110080, 2022.

\bibitem{marden:14}
J.~R. Marden and T.~Roughgarden, ``Generalized efficiency bounds in distributed
  resource allocation,'' {\em IEEE Transactions on Automatic Control}, vol.~59,
  no.~3, pp.~571--584, 2014.

\bibitem{dreves18}
A.~Dreves, ``How to select a solution in generalized {Nash} equilibrium
  problems,'' {\em Journal of Optimization Theory and Applications}, vol.~178,
  no.~3, pp.~973--997, 2018.

\bibitem{dreves19}
A.~Dreves, ``An algorithm for equilibrium selection in generalized {Nash}
  equilibrium problems,'' {\em Computational Optimization and Applications},
  vol.~73, no.~3, pp.~821--837, 2019.

\bibitem{Simonetto2020}
A.~Simonetto, E.~Dall'Anese, S.~Paternain, G.~Leus, and G.~B. Giannakis,
  ``Time-varying convex optimization: Time-structured algorithms and
  applications,'' {\em Proceedings of the IEEE}, vol.~108, no.~11,
  pp.~2032--2048, 2020.

\bibitem{lu2020}
K.~Lu, G.~Li, and L.~Wang, ``Online distributed algorithms for seeking
  generalized {Nash} equilibria in dynamic environments,'' {\em IEEE
  Transactions on Automatic Control}, vol.~66, no.~5, pp.~2289--2296, 2021.

\bibitem{meng21}
M.~Meng, X.~Li, Y.~Hong, J.~Chen, and L.~Wang, ``Decentralized online learning
  for noncooperative games in dynamic environments,'' {\em arXiv preprint,
  https://arxiv.org/abs/2105.06200}, 2021.

\bibitem{facchinei07}
F.~Facchinei and J.-S. Pang, {\em Finite-dimensional variational inequalities
  and complementarity problems}.
\newblock Springer, 2007.

\bibitem{bauschke11}
H.~H. Bauschke and P.~L. Combettes, {\em Convex analysis and monotone operator
  theory in {Hilbert} spaces}.
\newblock Springer, 2011.

\bibitem{yamada05}
I.~Yamada and N.~Ogura, ``Hybrid steepest descent method for variational
  inequality problem over the fixed point set of certain quasi-nonexpansive
  mappings,'' {\em Numerical Functional Analysis and Optimization}, vol.~25,
  no.~7-8, pp.~619--655, 2005.

\bibitem{xu03}
H.~Xu and T.~Kim, ``Convergence of hybrid steepest-descent methods for
  variational inequalities,'' {\em Journal of Optimization Theory and
  Applications}, vol.~119, no.~1, pp.~185--201, 2003.

\bibitem{cegielski14}
A.~Cegielski and R.~Zalas, ``Properties of a class of approximately shrinking
  operators and their applications,'' {\em Fixed Point Theory}, vol.~15, no.~2,
  pp.~399--426, 2014.

\bibitem{Bastianello2020}
N.~Bastianello, A.~Simonetto, and R.~Carli, ``Primal and dual
  prediction-correction methods for time-varying convex optimization,'' {\em
  arXiv preprint, arXiv:2004.11709}, 2020.

\bibitem{Simonetto17}
A.~Simonetto, ``Time varying convex optimization via time-varying averaged
  operators,'' {\em arXiv preprint, arXiv:1704.07338}, 2017.

\bibitem{Cegielski13}
A.~Cegielski, A.~Gibali, S.~Reich, and R.~Zalas, ``An algorithm for solving the
  variational inequality problem over the fixed point set of a
  quasi-nonexpansive operator in {Euclidean} space,'' {\em Numerical Functional
  Analysis and Optimization}, vol.~34, no.~10, pp.~1067--1096, 2013.

\bibitem{palomar10}
D.~P. Palomar and Y.~C. Eldar, {\em Convex optimization in signal processing
  and communications}.
\newblock Cambridge university press, 2010.

\bibitem{auslender00}
A.~Auslender and M.~Teboulle, ``Lagrangian duality and related multiplier
  methods for variational inequality problems,'' {\em SIAM Journal on
  Optimization}, vol.~10, no.~4, pp.~1097--1115, 2000.

\bibitem{malitsky20}
Y.~Malitsky and M.~K. Tam, ``A forward-backward splitting method for monotone
  inclusions without cocoercivity,'' {\em SIAM journal on optimization},
  vol.~30, no.~2, pp.~1451--1472, 2020.

\bibitem{tseng00}
P.~Tseng, ``A modified forward-backward splitting method for maximal monotone
  mappings,'' {\em SIAM Journal on Control and Optimization}, vol.~38, no.~2,
  pp.~431--446, 2000.

\bibitem{Dontchev14}
A.~L. Dontchev and R.~T. Rockafellar, {\em Implicit Functions and Solution
  Mappings}.
\newblock Springer Science \& Business Media, 2014.

\bibitem{zinkevich03}
M.~Zinkevich, ``Online convex programming and generalized infinitesimal
  gradient ascent,'' {\em Proceedings of the 20th international conference on
  machine learning}, 2003.

\bibitem{DallAnese2016}
E.~Dall'Anese and A.~Simonetto, ``Optimal power flow pursuit,'' {\em IEEE
  Transactions on Smart Grid}, 2016.

\bibitem{ogura03}
N.~Ogura and I.~Yamada, ``Nonstrictly convex minimization over the bounded
  fixed point set of a nonexpansive mapping,'' {\em Numerical Functional
  Analysis and Optimization}, vol.~24, no.~1-2, pp.~129--135, 2003.

\end{thebibliography}
	\bibliographystyle{ieeetr}

\balance
\end{document}